\DeclareUnicodeCharacter{2212}{-}
\documentclass[11pt]{article} 

\usepackage[backend=bibtex,style=alphabetic,backref=true,,maxnames=99]{biblatex}
\usepackage[letterpaper, margin=1in]{geometry} 

\title{Modifications of Quantum Computation and\\ Adaptive Queries to $\mathsf{PP}$\footnotetext{This work is supported by US
Department of Energy (grant no DE-SC0023179) and partially supported by US National Science Foundation
(award no 1954311).}\footnotetext{We thank Dorian Rudolph, Yuki Takeuchi and Thomas Huffstutler for helpful discussions and the reviewers for helpful comments.}}

\author{David Miloschewsky\\\small{Department of Computer Science, }\\
\small{Stony Brook University}\\
\small{dmiloschewsk@cs.stonybrook.edu}\and Supartha Podder
\\
\small{Department of Computer Science, }\\
\small{Stony Brook University}\\\small{supartha@cs.stonybrook.edu}}

\date{}

\usepackage{amssymb}
\usepackage{amsmath}
\usepackage{amsthm}
\usepackage{xspace}
\usepackage{physics}
\usepackage{xcolor}
\usepackage{tikz}
\usetikzlibrary{arrows.meta,positioning,calc,matrix,fit,bending}
\usepackage{quantikz}
\usepackage{enumitem}
\usepackage{mathtools}
\usepackage{algpseudocode}
\usepackage{diagbox}
\usepackage{hyperref}
\usepackage[capitalise,noabbrev,nameinlink]{cleveref}
\usepackage{multirow}
\usepackage{dirtytalk}

\newcommand{\compclass}[2][]{\ensuremath{\mathsf{#2}^{#1}}\xspace}
\newcommand{\pdqp}[0]{\compclass{PDQP}}
\newcommand{\corrbqp}[1][]{\compclass{CorrBQP_{#1}}}
\newcommand{\postbqp}[0]{\compclass{PostBQP}}
\newcommand{\bqp}{\compclass{BQP}}

\newcommand{\cocorrbqp}[0]{\compclass{coCorrBQP}}
\newcommand{\corrbqpsc}[2]{\compclass{CorrBQP(#1, #2)}}
\newcommand{\corrbqporacle}[1]{\ensuremath{\compclass[#1]{CorrBQP}_{\text{classical}}}}

\newcommand{\countinghierarchy}[0]{\compclass{CH}}
\newcommand{\countinghierarchylevel}[1]{\compclass{C_{#1}P}}
\newcommand{\pp}[0]{\compclass{PP}}
\newcommand{\sharpp}[0]{\ensuremath{\compclass{\# P}}\xspace}
\newcommand{\bpp}[0]{\compclass{BPP}}
\newcommand{\posteqp}[0]{\compclass{PostEQP}}
\newcommand{\adposteqp}[0]{\compclass{AdPostEQP}}
\newcommand{\postbqpclassical}[0]{\ensuremath{\postbqp_{\textup{classical}}}\xspace}
\newcommand{\postbpp}[0]{\compclass{PostBPP}}
\newcommand{\pspace}[0]{\compclass{PSPACE}}

\newcommand{\majbqp}{\compclass{MajBQP}}
\newcommand{\admajbqp}{\compclass{AdMajBQP}}
\newcommand{\eqp}{\compclass{EQP}}
\newcommand{\rweqp}{\compclass{RwEQP}}
\newcommand{\pdeqp}{\compclass{PDEQP}}
\newcommand{\ceqp}{\compclass{CEQP}}
\newcommand{\correqp}{\compclass{CorrEQP}}
\newcommand{\majeqp}{\compclass{MajEQP}}
\newcommand{\pqp}{\compclass{PQP}}
\newcommand{\adpdqp}{\compclass{AdPDQP}}

\newcommand{\adpostbqp}[0]{\compclass{AdPostBQP}}

\newcommand{\ratsumdeg}[1][]{\ensuremath{\text{trdeg}_{#1}}\xspace}
\newcommand{\adpostq}[1]{\ensuremath{\text{AdPostQ}_{#1}}}
\newcommand{\postq}[1]{\ensuremath{\text{PostQ}_{#1}}}
\newcommand{\ratdeg}[1][]{\ensuremath{\text{rdeg}_{#1}}\xspace}
\DeclareMathOperator{\E}{\mathbb{E}}

\newtheorem{theorem}{Theorem}[section]
\newtheorem{corollary}[theorem]{Corollary}
\newtheorem{lemma}[theorem]{Lemma}
\newtheorem{definition}[theorem]{Definition}

\newtheorem{observation}[theorem]{Observation}

\newtheorem{claim}[theorem]{Claim}

\newcommand\numberthis{\addtocounter{equation}{1}\tag{\theequation}}

\newcommand{\adaf}[2]{\ensuremath{\textup{Ada}_{#1, #2}}}

\newcommand{\res}[1]{\hyperref[result:#1]{Result~\ref*{result:#1}}}
\newcommand{\tabl}[1]{\hyperref[table:#1]{Table~\ref*{table:#1}}}
\definecolor{darkred}{RGB}{139, 0, 0}
\begin{document}

\maketitle

\begin{abstract}
In 2004, Aaronson~\cite{pp_postbqp} introduced the complexity class $\mathsf{PostBQP}$~($\mathsf{BQP}$ with postselection) and showed that it is equal to $\mathsf{PP}$. 
Following their line of work, we introduce two new complexity classes. The first, $\mathsf{CorrBQP}$, is a modification of $\mathsf{BQP}$ which has the power to perform \emph{correlated measurements}, i.e. measurements that output the same value across a partition of registers. The second, $\mathsf{MajBQP}$, augments $\mathsf{BQP}$ with the ability to collapse a register to its most likely measurement outcome. Specifically, we consider two variants, $\mathsf{MajBQP}$ and $\mathsf{AdMajBQP}$, where the latter may perform intermediate measurements. We exactly characterize the computational power of the models,
\begin{align*}
    &\mathsf{CorrBQP} = \mathsf{AdMajBQP} = \mathsf{BPP}^{\mathsf{PP}} &\mathsf{MajBQP} = \mathsf{P}^{\mathsf{PP}}.
\end{align*}

In fact, we show that other metaphysical modifications of $\mathsf{BQP}$, such as $\mathsf{CBQP}$ (i.e. $\mathsf{BQP}$ with the ability to clone arbitrary quantum states, introduced in~\cite{grover_search_no_signaling_principle, rewindable_bqp}), are also equal to $\mathsf{BPP}^{\mathsf{PP}}$. We show that $\mathsf{CorrBQP}$ and $\mathsf{MajBQP}$ are self-low with respect to classically-accessible queries. In contrast, if they were self-low under quantumly-accessible queries, the counting hierarchy would collapse. 
Furthermore, we introduce a variant of rational degree that lower-bounds the query complexity of $\mathsf{BPP}^{\mathsf{PP}}$. Lastly, we extend the adversary lower-bounding technique to $\mathsf{AdPDQP}$, $\mathsf{BQP}$ with the ability to sample the current state of an algorithm with collapsing it and adapt the computation based on the samples.
\end{abstract}

\section{Introduction}

The study of quantum computing spurred the study of various quantum complexity classes.
Starting with the introduction of \bqp, the quantum analogue of \compclass{P}, numerous quantum counterparts to standard classical complexity classes, such as \compclass{QMA}, \compclass{QIP} and \compclass{QAC_0} (which respectively correspond to \compclass{NP}, \compclass{IP}, and \compclass{AC_0}), have been defined. 
One question about these classes is how they relate to their classical counterparts and whether such connections can reveal new structural insights or lead to progress on foundational questions in computational complexity.

One such results is $\postbqp = \pp$~\cite{pp_postbqp}, where \postbqp denotes \bqp with the ability to postselect on any specific measurement outcome, even that occurs with exponentially small probability. This result led to a much simpler proof that \pp is closed under intersection and shows how exploring alternative formulations of \bqp can potentially yield new insights into the structure and limitations of other complexity classes. 

Our work contributes to this exploration by introducing two modifications of \bqp. We find that our models exactly capture computation with access to a \pp oracle, address the space between \pp and \pspace, and discuss the power of classically-accessible queries. Next, let us describe the two models we introduce.

\subsection{Models}

Let us shortly describe the computational models we study. A formal definition of these models may be found in~\cref{section:computational_models}.

\subsubsection{Correlated Measurements}

\begin{figure}
    \centering
    \resizebox{\linewidth}{!}{\begin{tikzpicture}[
  font=\Large,
  >=Latex,
  arr/.style={->, thick},
  box/.style={draw, rounded corners=2pt, inner sep=7pt, align=center},
  num/.style={font=\large, fill=white, inner sep=1pt},
]

\node[box, text width=0.45\linewidth] (in) {$
\ket{\psi}
=
\Big(\sqrt{\tfrac23}\ket{0}_L+\sqrt{\tfrac13}\ket{1}_L\Big)
\otimes
\Big(\sqrt{\tfrac15}\ket{0}_F+\sqrt{\tfrac45}\ket{1}_F\Big)
$};

\node[box, right=12mm of in, text width=0.45\linewidth] (exp) {$
\begin{aligned}
\ket{\psi}=\;&
{\color{green!45!black}\sqrt{\tfrac{2}{15}}\ket{00}}
+{\color{gray!65}\sqrt{\tfrac{8}{15}}\ket{01}}\\
&+
{\color{gray!65}\sqrt{\tfrac{1}{15}}\ket{10}}
+{\color{green!45!black}\sqrt{\tfrac{4}{15}}\ket{11}}
\end{aligned}
$};

\draw[arr] (in.east) -- (exp.west);
\node[box, below=16mm of exp, xshift=-55mm, text width=0.42\linewidth] (keep) {$
\ket{\psi_{\mathrm{keep}}}
=
{\color{green!45!black}\sqrt{\tfrac{2}{15}}\ket{00}}
+
{\color{green!45!black}\sqrt{\tfrac{4}{15}}\ket{11}}
$};

\node[box, right=12mm of keep, text width=0.34\linewidth] (rw) {$
\ket{\psi_{\mathrm{bal}}}
=
{\color{green!45!black}\sqrt{\tfrac23}\ket{00}}
+
{\color{green!45!black}\sqrt{\tfrac13}\ket{11}}
$};

\node[inner sep=0pt, right=10mm of rw] (meas) {
  \begin{quantikz}[column sep=1mm, row sep=0mm]
    & \meter{} \qw
  \end{quantikz}
};

\node[box, right=10mm of meas, yshift=8mm, text width=0.26\linewidth] (o0) {$\ket{00}\ \text{w.p. }\tfrac23$};
\node[box, right=10mm of meas, yshift=-8mm, text width=0.26\linewidth] (o1) {$\ket{11}\ \text{w.p. }\tfrac13$};

\draw[arr] (exp.south) -- ++(0,-6mm) -| node[num, pos=0.55] {(1)} (keep.north);
\draw[arr] (keep.east) -- node[num, midway] {(2)} (rw.west);
\draw[arr] (rw.east) -- (meas.west);
\draw[arr] (meas.east) -- (o0.west);
\draw[arr] (meas.east) -- (o1.west);

\end{tikzpicture}}
    \caption{Visualization of a correlated measurement where the register $L$ is the leader and $F$ is a follower. Step (1) represents the conditioning on \emph{correlated} values in the registers, while Step (2) represents the weight redistribution based on the leader register $L$.}
    \label{fig:model_explanation}
\end{figure}

First, we introduce \corrbqp, a modified version of \bqp enhanced with the power to perform \emph{correlated measurements}. In summary, a correlated measurement is an operation where we obtain the same measurement outcome across multiple registers and the probability distribution of each outcome is dictated by a specific register. A visualization of this process may be seen in~\cref{fig:model_explanation}.

In more detail, a correlated measurement is the operation which takes two inputs, a state $\ket{\psi}$ and a partition over a subset registers $p = (p_1, \dots p_k)$ where $\forall i, j\in [k]$, $\abs{p_i}=\abs{p_j}$. Note that we are assuming the both the partitions and the registers within each partition are ordered. The first partitioned set $p_1$ is denoted as the \emph{leader set}, while the rest are \emph{followers}. In~\cref{fig:model_explanation}, the registers $L$ denote the leader set, while $F$ denotes the followers. The correlated measurement proceeds in three steps. First, we \say{discard} any weight of $\ket{\psi}$ over registers with different values (Step~(1) in~\cref{fig:model_explanation}). Next, we adjust the weight distribution in order to follow the leader register (Step~(2)). Finally, we measure the state.

At first, \corrbqp might \say{trivially} appear contained with \postbqp, but that is false. It is true that we may obtain the same measurement outputs using postselection (Simulate Step (1)). The caveat is that we don't know how to produce the same state as in Step (2). Suppose that we are performing a correlated measurement on $k$ states, each being $\alpha \ket{0} + \beta\ket{1}$. We will observe $\ket{0}$ with probability $\abs{\alpha}^2$. On the other hand, if we postselect on obtaining the same output across all $k$ registers and measure afterwards, we will measure $\ket{0}$ with probability $\abs{\alpha^k}^2/(\abs{\alpha^k}^2 + \abs{\beta^k}^2)$. However, \corrbqp may simulate \postbqp by setting the leader register to $\ket{0}$.

\subsubsection{Quantum Majority}

Second, we introduce \majbqp, \bqp with the ability to apply \emph{quantum majority} gates. A quantum majority gate is a single-qubit gate which, when applied to the state $\ket{\psi}=\alpha\ket{0} + \beta\ket{1}$, collapses to $\ket{0}$ if $\abs{\alpha}^2 \geq \abs{\beta}^2$ and $\ket{1}$ otherwise. Note that we normalize the state after the application of the majority gate, similar to postselection. Similarly to \corrbqp, we may observe that $\postbqp \subseteq \majbqp$. This is due to the equivalence $\postbqp = \compclass{PQP}$~\cite{watrous_quantum_computational_complexity} (\compclass{PQP} is \bqp with unbounded error) and the fact that we may apply a quantum majority gate at the end of the computation. Lastly, we consider two versions of the class based on partial measurements. If the quantum circuit may perform intermediate measurements, we call the class \admajbqp, while if it cannot, we denote it by \majbqp.

\subsection{Results}

\begin{figure}
    \centering
    \begin{tikzpicture}
    \node (bqp) at (3,-2) {\bqp};
    \node (pp) at (2,0) {$\compclass{PP}=\postbqp =\compclass{PQP}$};
    \node (ph) at (6,0) {\compclass{PH}};
    \node (p_pp) at (3, 1.5) {$\compclass[\compclass{PP}]{P}\color{red}=\color{black} \compclass{MajBQP}$};
    \node (qcph) at (6, 3.2) {\compclass{QCPH}};
    \node (bpp_pp) at (2,3.2) {$\compclass[\compclass{PP}]{BPP} \color{red}=\color{black} \compclass{CorrBQP} \color{red}=\color{black}  \admajbqp{}^\dagger$};
    \node (pspace) at (3.5, 6.3) {\compclass{PSPACE}};
    \node (pdqp) at (-2,0) {\compclass{PDQP}};
    \node (adpdqp) at (0,1.5) {\compclass{AdPDQP}};
    \node (countinghierarchy) at (3.8,5) {\countinghierarchy};

     \foreach \source/\target in {bqp/pp, bqp/qcph, bqp/pdqp, pp/p_pp, ph/p_pp, ph/qcph, p_pp/bpp_pp, qcph/countinghierarchy, bpp_pp/countinghierarchy, countinghierarchy/pspace, pdqp/adpdqp,  adpdqp/bpp_pp}
    \draw[->] (\source) -- (\target);
\end{tikzpicture}
    \caption{A map of relationships between the complexity classes discussed. Given two classes $C$ and $D$, $C\rightarrow D$ indicates $C\subseteq D$. The \textcolor{red}{red} equal signs $\color{red}=$ signify our results. ${}^\dagger$See Result 1 for all classes equal to $\compclass[\pp]{BPP}$.}
    \label{fig_complexity_class_spider}
\end{figure}

Let us describe our results.

\begin{quote}
    \textbf{Result 1:} We show that $\majbqp = \compclass[\pp]{P} $ and $\corrbqp = \admajbqp = \compclass[\pp]{BPP} = \compclass{CBQP} = \adpostbqp = \compclass{RwBQP}$(\cref{definition:adpostbqp};~\cref{thm:equality_theorem}).
\end{quote}

We begin by exactly characterizing \corrbqp, \majbqp and \admajbqp as computation with access to a \pp oracle. We find that the same characterization holds for other modifications of \bqp, such as \compclass{CBQP}, \adpostbqp, and \compclass{RwBQP}.\footnote{\compclass{CBQP} is \bqp with the ability to copy any state;  \adpostbqp is \postbqp with the ability to perform intermediate measurements; \compclass{RwBQP} is \bqp with the ability to uncompute measurements.} As the best prior upper-bounds for these classes were \pspace, and \compclass[\pp]{BPP} is not believed to be equal to \pspace, this result indicates that their computational power is only slightly stronger than \postbqp.
Furthermore, we find that \corrbqp is self-low with respect to classical queries, a quality which \postbqp is not known to possess. Additionally, to build intuition for \corrbqp, we describe a $O(\log n)$-query algorithm to solve Parity (\cref{lem:corrbqp_parity}). This is in contrast to \postbqp which requires $\Omega(n)$ queries to solve the problem~\cite{MdW14,IJK+25}. Finally, Result 1 highlights the power of intermediate measurements. In particular, the computational gaps between \admajbqp and \majbqp, and between \adpostbqp and \postbqp, are naturally explained by the presence or absence of intermediate measurements.

We use this result to study the counting hierarchy, \countinghierarchy. It is defined as $\countinghierarchy = \cup_{i\in\mathbb{N}} \countinghierarchylevel{i}$ where $\countinghierarchylevel{1} = \pp$ and $\countinghierarchylevel{i} = \compclass[\countinghierarchylevel{i-1}]{PP}$.\footnote{Alternatively, one may look at \countinghierarchy as a version of the polynomial hierarchy, \compclass{PH}, where the $\exists$ and $\forall$ quantifiers are replaced by a counting quantifier $\mathcal{C}$, a quantifier which asks whether the ratio of computational paths an \compclass{NP} machine accepts is over $1/2$.} 
Two reasons motivate  the study of \countinghierarchy. The first is the seminal result of Toda which shows that $\compclass{PH}\subseteq \compclass[\pp]{P}$, placing it between the first and second levels of \countinghierarchy~\cite{toda_theorem}. Secondly, a variety of interesting problems are complete for levels of \countinghierarchy. Specifically, $\textup{MAJ-MAJ-SAT}$ is $\compclass[\pp]{PP}$-complete, $\textup{E-MAJ-SAT}$ is $\compclass[\pp]{NP}$-complete, both of which are relevant in AI literature~\cite{LGM98,CXD12, OCD16}, and~\cite{Tod94} shows that computing the middle-bit of any \compclass{GapP} function is $\compclass[\pp]{P}$-complete.

\begin{quote}
    \textbf{Result 2:} If \corrbqp is self-low with respect to quantum queries, i.e. $\compclass[\corrbqp]{CorrBQP}\subseteq \corrbqp$, then $\countinghierarchy=\compclass[\pp]{BPP}$ (\cref{lemma_ch_collapse_implications}).
\end{quote}

Assuming \countinghierarchy does not collapse, this result implies the existence of a class which is self-low with respect to classically-accessible queries, but not quantumly-accessible queries. The behavior of various quantum query models has been studied extensively in complexity theory. For example, prior work has examined both the distinction between classical and quantum oracles~\cite{AK07, Aar09, cautionary_note_on_quantum_oracles} and the effect of the access model, such as whether oracle queries may be made coherently in superposition~\cite{YZ24, LLPY23}. Our result supports the trend that relativizing proof techniques may behave differently depending on whether the oracle is classically or quantum accessible. Next, we study alterations of \bqp in the query complexity setting.

\begin{quote}
    \textbf{Result 3:} We define the measure \emph{rational tree degree} (\cref{definition:rational_tree_degree}) which lower-bounds the query complexity of \adpostbqp (\cref{theorem_trdeg_lowerbounds_adpostq}).
\end{quote}

We define a new variant of rational degree (\ratdeg), which we call \emph{rational tree degree} (\ratsumdeg), and show that it lower-bounds the query complexity of \adpostbqp. We find that in the zero error case, the measures are equal, while in the constant error case, we are able to separate them. Two reasons motivate the study of \ratsumdeg. Firstly, we want to expand on the prior work which connected rational degree and \postbqp~\cite{MdW14}. Secondly, due to the fact that $\compclass[\pp]{BPP}\subseteq \adpostbqp$ relativizes, \ratsumdeg may be used as a lower-bounding technique for \compclass[\pp]{BPP}, potentially leading to an oracle separation between \compclass[\pp]{BPP} and \pspace.

Furthermore, we explore what happens when the query access of a quantum class to an oracle is restricted to classical queries. Specifically, we place this restriction on \postbqp, obtaining the class \postbqpclassical.

\begin{quote}
    \textbf{Result 4:} There exists an oracle $O$ such that $\compclass[O]{BQP}\not\subseteq \compclass[O]{PostBQP}_{\textup{classical}}$ and $\compclass[O]{PP}\not\subseteq \compclass[O]{PostBQP}_{\textup{classical}}$ (\cref{theorem:separate_bqp_postbqp_classical}) where \postbqpclassical is \postbqp restricted to classical queries.
\end{quote}

Our motivation for this is due to the fact that the equality $\pp = \postbqp$ holds relative to \emph{every} quantum accessible classical oracle~\cite{pp_postbqp, adaptivityvspostselection}. One would wonder if the same would hold true for \emph{every} classically accessible oracles.\footnote{In fact it turns out that, if $\compclass[O]{PP}\subseteq \compclass[O]{PostBQP}_{\textup{classical}}$ for all oracles $O$, then \countinghierarchy collapses to \compclass[\pp]{P}.}  
This result shows that this is indeed false. We do this by extending the oracle separation between \bqp (with quantum queries) and \postbpp.\footnote{\postbpp, also known as $\compclass{BPP}_{\mathsf{Path}}$, was defined in~\cite{postbpp} as the variant of \compclass{BPP} with the ability to postselect on some event occurring.} As $\bqp\subseteq \pp$ relativizes, our result follows.

Lastly, we study the query complexity of the adaptive version of \pdqp, denoted \adpdqp. Introduced in~\cite{space_above_bqp}, \pdqp augments \bqp with non-collapsing measurements. Informally, the \bqp algorithm may sample its current state without collapsing it, but all samples are processed at the end of the computation. Consequently, the sequence of unitaries in \pdqp may be fixed in advance. In contrast, the unitaries in \adpdqp may depend on the outcomes of earlier non-collapsing measurements, meaning this could strengthen the computational power of the model; see~\cref{fig:adpdqp_original}. Since our paper studies adaptive versions of \postbqp and \majbqp, \adpdqp is a natural model to examine as well. We show, however, that in this setting adaptivity does not yield a major advantage, resolving an open question from~\cite{space_above_bqp}.

\begin{quote}
   \textbf{Result 5:} We extend the adversary method to \adpdqp (\cref{thm:adversary_bound_adpdqp}). We find that for any \adpdqp query algorithm solving unstructured search requires $\Omega(n^{1/4})$ queries and parity requires $\Omega(n^{1/2})$ queries. Therefore, there exist oracles $O_1, O_2$ such that $\compclass[O_1]{NP}\not\subseteq \compclass[O_1]{AdPDQP}$ and $\compclass[O_2]{\oplus P}\not\subseteq \compclass[O_2]{AdPDQP}$.
\end{quote}

We find this result surprising as adding adaptivity to \postbqp allowed us to compute parity efficiently, while adding adaptivity to \pdqp does not aid with search nor parity. Our argument is an extension of the proof of the adversary bound for \pdqp~\cite{revisiting_bqp_with_noncollapsing_measurements}, with our main contribution being a simulation of each non-collapsing measurement which is both adaptive and does not entangle multiple simulations.

\subsection{Related Work}

The study of modifications of \bqp has a long history. To the best of our knowledge, the first investigation of metaphysical quantum computation was through the study closed timelike curves (CTCs)~\cite{ctc_deutsch}, whose computational power is equal to \pspace~\cite{ctc_pspace}, making it essentially the strongest known modification. For example, it has been shown that CTCs would allow us to clone arbitrary states~\cite{ctc_cloning}.

The power of cloning arbitrary states was described in~\cite{grover_search_no_signaling_principle}, where it was shown that it can lead to an exponential speed-up in solving any search problem. The complexity class associated with this power, \compclass{CBQP}, was formally studied in~\cite{rewindable_bqp}, where it was shown that its power is equal to \compclass{RwBQP}, and \adpostbqp.

The power of postselection was initialized in~\cite{pp_postbqp}, where it was shown that $\postbqp = \pp$. Furthermore, it has been shown that it may solve any boolean problem when combined with quantum polynomial advice, formally meaning that $\compclass{PostBQP/qpoly} = \compclass{ALL}$~\cite{limitations_of_quantum_advice}. Additionally, \postbqp was considered in the query setting as well, where it was shown that its query complexity is tightly described by the measure \ratdeg, the degree of a rational function~\cite{MdW14}.

Arguably, weakest (currently known) metaphysical ability which adds computational power is non-collapsing measurements, the ability to obtain a measurement sample without altering the state. It was shown that the associated complexity class \pdqp sits between \compclass{SZK} and \compclass[\pp]{BPP}~\cite{space_above_bqp}. However, its power increases significantly when combined with quantum advice, yielding $\compclass{PDQP/qpoly} = \compclass{ALL}$~\cite{pdqpqpoly}, or when combined with \compclass{QMA}, in which case $\compclass{PDQMA} = \compclass{NEXP}$~\cite{pdqma_nexp_1, pdqma_nexp_a}. Additionally, it is interesting to note that \pdqp and \adpostbqp gain additional power from the presence of intermediate measurements, without them, these classes collapse to \bqp and \postbqp, respectively. In contrast, classes such as \bqp or \corrbqp, do not exhibit such behavior.

Lastly, several works consider the space between \pp and \pspace. Toda's theorem shows that $\compclass{PH}\subseteq \compclass[\pp]{P}$~\cite{toda_theorem}. A similar result has been shown for a quantum version of \compclass{PH}, $\compclass{QCPH}\subseteq \compclass[\mathsf{PP}^{\mathsf{PP}}]{P}$~\cite{qph_original}. Alternatively, several papers discuss implications which would collapse the counting hierarchy~\cite{bqpqpoly_pp,bqpqpoly_pp_yirka}.
On the oracle side, few separations are known. In particular, there exist oracles $O$ such that $\compclass[O]{PP}\not \subseteq \compclass[\mathsf{NP}^{O}]{P}$~\cite{perceptrons_pp_ph, adaptivityvspostselection} and $\compclass[O]{\oplus P}\not \subseteq \compclass[\mathsf{PH}^O]{PP}$~\cite{parity_p_oracle_pp_ph}.

\section{Preliminaries}

In this section we list the definitions we use. We assume familiarity with standard quantum computing and computational complexity. For a deeper introduction, see~\cite{Nielsen_Chuang} and~\cite{arora_barak}.

\subsection{Notation}

We begin by explaining several points about the notation. The set $\{1,...n\}$ is denoted by $[n]$. Given a string $x\in\{0,1\}^n$, $x_i$ for $i\in [n]$ denotes the value of $x$ at index $i$. However, if the subscript is already used for a different purpose, we denote the value using $x(i)$. The complexity class $\mathcal{C}^\mathcal{O}$ means a complexity class $\mathcal{C}$ with oracle access to $\mathcal{O}$. Given a tower of oracles, such as $\mathcal{C}^{\mathcal{O}_1^{...^{\mathcal{O}_n}}}$, we assume that each level only has direct access to the oracle one level above, giving them indirect access to all oracles above. If $\mathcal{C}$ is a complexity class defined using a quantum TM, then we assume that we may make quantum queries to $\mathcal{O}$ (i.e. queries in superposition). If we want to restrict it to classical queries to $O$, meaning that every query input is a classical string, we label it as $\mathcal{C}^\mathcal{O}_{\textup{classical}}$.\footnote{This may be enforced in different ways. If the base machine may perform intermediate measurements, without loss of generality we may assume the query register is always measured before the query. If the base machine can not measure during the computation, we enforce that the state is in a basis state before a query.}

\subsection{Complexity Classes}

We provide the definitions of the complexity classes we use.

\begin{definition}[\pp and \pqp]\label{definition_computational_class_pp}
    A language $L$ is said to be in \pp if there exists a probabilistic Turing Machine $M$ such that for all inputs $x$,
    \begin{itemize}
        \item If $x\in L$, then $\Pr(M(x) = 1) > 1/2$.
        \item If $x \not \in L$, then $\Pr(M(x) = 1) \leq 1/2$.
    \end{itemize}
    Furthermore, the class \pqp is defined equivalently, except we use a quantum Turing Machine.
\end{definition}

Essentially, \pp is the probabilistic polynomial-time class where the gap between completeness and soundness may be arbitrarily small.

\begin{definition}[\postbqp,~\cite{pp_postbqp}]\label{definition_computational_class_postbqp}
    A language $L$ is said to be in \postbqp if there exists a uniform family of polynomial-sized quantum circuits $\{C_n\}_{n\geq 1}$ such that for all inputs $x$,
    \begin{itemize}
        \item After applying $C_n$ to $\ket{0...0}\otimes\ket{x}$, the first register has at least $1/\exp(n)$ probability of being measured as $\ket{1}$.
        \item If $x\in L$, then conditioned on obtaining $\ket{1}$ in the first register, the probability the second register is $\ket{1}$ is at least $2/3$.
        \item If $x\not \in L$, then conditioned on obtaining $\ket{1}$ in the first register, the probability the second register is $\ket{1}$ is at most $1/3$.
    \end{itemize}
    Without loss of generality, we assume that $C_n$ is composed of Hadamard and Toffoli gates.
\end{definition}

\postbqp is the class of languages which can be decided using postselection. Here, the first register is the one we are postselecting on, while the second register indicates containment of the input in the language. It suffices to postselect on one register as we can always entangle it with other registers. It has been shown that the two definitions above are equivalent and this result relativizes.

\begin{theorem}[\cite{pp_postbqp},~\cite{watrous_quantum_computational_complexity}]\label{theorem:pp_equals_postbqp}
    With respect to any classical oracle $O$,
    \begin{align*}
        \compclass[O]{PP} = \compclass[O]{PostBQP} = \compclass[O]{PQP}
    \end{align*}
\end{theorem}

Next, we define the counting hierarchy.

\begin{definition}[Counting Hierarchy]

    We define the complexity class \countinghierarchylevel{k}, known as the kth level of the counting hierarchy, as:
    \begin{itemize}
        \item $\countinghierarchylevel{0} = \compclass{P}$
        \item $\countinghierarchylevel{1} = \pp$
        \item $\forall k > 1, \countinghierarchylevel{k} = \compclass[\countinghierarchylevel{k-1}]{PP}$
    \end{itemize}
    Their union over constant $k$ is the counting hierarchy, $\countinghierarchy = \cup_{k\in\mathbb{N}} \countinghierarchylevel{k}$.
\end{definition}

Lastly, we define an adaptive variation of \postbqp known as \adpostbqp (due to~\cite{rewindable_bqp}). Before we do so, let us define a quantum circuit with postselection and intermediate measurements.

\begin{definition}[Quantum Circuit with Postselection and Intermediate Measurements]\label{definition_circuit_postselect_measurements}
    A \emph{quantum circuit with postselection and intermediate measurements} is a quantum circuit $C_n$ over a fixed universal gate set which is allowed to use,
    \begin{itemize}
        \item single-qubit measurements operators $M$, and
        \item single qubit postselection operators $\ketbra{1}{1}$.
    \end{itemize}
    For a fixed input $x$ and circuit $C_n$, let $\mathcal{A}$ denote the set of partial measurement outcomes of $C_n\ket{x}$ and let $p$ be the event that all postselection operators succeed. The probability $C_n$ outputs $o$ is described as,
    \begin{align*}
        \sum_{a\in\mathcal{A}} \Pr[a]\cdot \Pr[C_n \ket{x} = o \mid \text{$p$ and $a$}]
    \end{align*}
    where we assume that for all $a\in\mathcal{A}$ such that $\Pr[a]>0$, $\Pr[p\mid a]>0$.
\end{definition}

Essentially, the projector $\ketbra{1}{1}$ is the one performing postselection in the circuit. The additional power we obtain, in \adpostbqp, over standard \postbqp is that we may perform intermediate measurements during the computation. For example, based on a partial measurement result, we may decide whether to apply a $X$ gate before the projector, effectively changing which value we postselect on.

\begin{definition}[\adpostbqp \cite{rewindable_bqp}]\label{definition:adpostbqp}
    A language $L$ is said to be in \adpostbqp if there exists a family of \emph{quantum circuits with postselection and intermediate measurements} $\{C_n\}_n$, defined in~\cref{definition_circuit_postselect_measurements}, such that,
    \begin{itemize}
        \item If $x\in L$, $C_n$ outputs 1 with probability at least $2/3$.
        \item If $x\not \in L$, $C_n$ outputs 1 with probability at most $1/3$.
    \end{itemize}
\end{definition}

It has been shown that this type of computation is equivalent to other modifications of quantum computation. Namely, quantum computation with the ability to copy arbitrary states (known as $\compclass{CBQP}$~\cite{rewindable_bqp, grover_search_no_signaling_principle}) and the ability to rewind to a state after any measurement (known as $\compclass{RwBQP}$). A formal definition of \compclass{CBQP} and \compclass{RwBQP} is in~\cite{rewindable_bqp}, who show the lemma below.

\begin{theorem}[\cite{rewindable_bqp}]\label{theorem_adpostbqp_equals_rwbqp_equal_cqbp}
    $\compclass[\pp]{BPP}\subseteq \adpostbqp = \compclass{RwBQP} = \compclass{CBQP} \subseteq \pspace$.
\end{theorem}

\subsection{Complexity Measures}

An $n$-variate polynomial is a function $P:\{0,1\}^n\rightarrow \mathbb{R}$ which may be written as $P(x_1,..,x_n) = \sum_{d_1,..,d_n} c_{d_1\dots d_n} \prod_{i\in [n]} x_i^{d_i}$. Its degree is $\deg(P) = \max\{\sum_{i\in[n]}d_i|c_{d_1\dots d_n}\neq 0\}$. Without loss of generality, as the domain is $\{0,1\}^n$, we may assume $P$ is a multilinear polynomial. This means that for all $i$, $d_i\in \{0,1\}$. We say a polynomial $P$ $\epsilon$-approximates a function $f$ if for all $x$, $\abs{P(x) - f(x)}\leq~\epsilon$.

A rational function is a ratio $R = P/Q$ of polynomials $P,Q: \{0,1\}^n\rightarrow \mathbb{R}$ where $\forall x\in\{0,1\}^n$, $Q(x)\neq 0$. The degree of a rational function is $\deg(R) = \max\{\deg(P), \deg(Q)\}$. The rational degree of a function $f$, denoted $\text{rdeg}_{\epsilon}(f)$, is the smallest degree over all rational functions $R$ which $\epsilon$-approximate $f$.

\cite{MdW14} showed that the query complexity of a \postbqp algorithm computing a function $f: \{0,1\}^n\rightarrow \{0,1\}$ up to $\epsilon$ error, denoted $\postq{\epsilon}(f)$, is equivalent to the $\ratdeg[\epsilon](f)$. Formally, they show,
\begin{align*}
    \frac{1}{2}\ratdeg[\epsilon](f) \leq \postq{\epsilon}(f) \leq \ratdeg[\epsilon](f) 
\end{align*}

Specifically, to show that $\postq{\epsilon}(f) \leq \ratdeg[\epsilon](f)$, they describe an algorithm with postselection which evaluates a rational function $R$.

\subsection{Problems}

\subsubsection{Forrelation}

The Forrelation problem was introduced in~\cite{bqp_and_ph} and was shown to optimally separate \bqp from \compclass{BPP}~\cite{forrelation}. For our purposes, we only need two characteristics of the problem. Let us begin with the following definition.

\begin{definition}[$\epsilon$-almost $k$-wise Equivalency,~\cite{bpppath_bqp_oracle_separation_fix}]
    Consider two distributions $D_0$, $D_1$ over $\{0,1\}^n$. We say that $D_0$ and $D_1$ are $\epsilon$-almost $k$-wise equal if for every $k$-term $C$,
    \begin{align*}
        1-\epsilon \leq \frac{\Pr_{D_0}[C]}{\Pr_{D_1}[C]} \leq 1+\epsilon
    \end{align*}
    where by a $k$-term, given a string $z\in\{0,1\}^n$, we mean a product of $k$ terms $z_i$ or $1-z_i$ for $i\in[n]$.
\end{definition}

We will use the theorem below.

\begin{theorem}\label{theorem_forrelation_k_wise}[Forrelation properties, adapted from~\cite{bqp_and_ph, bpppath_bqp_oracle_separation_fix}]

There exist two distributions $\mathcal{F}$ and $\mathcal{U}$ on $\{0,1\}^n$ such that:
\begin{itemize}
    \item Suppose that with probability $1/2$, we sample $x$ from $\mathcal{F}$ and otherwise we sample it from $\mathcal{U}$. Then a \bqp machine with a single query in superposition to $x$ may decide which is the case.
    \item For any $k=n^{o(1)}$, $\mathcal{F}$ and $\mathcal{U}$ are $o(1)$-almost $k$-wise equivalent.
\end{itemize}
\end{theorem}

\subsubsection{Adaptive Construction}

\cite{adaptivityvspostselection} introduced the adaptive construction of a boolean function.

\begin{definition}[Adaptive Construction]
    Given $f:\{0,1\}^n\rightarrow \{0,1\}$ and $d\in\mathbb{N}$, the adaptive construction of $f$ at depth $d$, $\adaf{f}{d}$, is defined as:
    \begin{align*}
        \adaf{f}{0} &= f \\
        \adaf{f}{d}(x,y,z) &= 
        \begin{cases}
            \adaf{f}{d-1}(y) &\text{ if }f(x)=0\\
            \adaf{f}{d-1}(z) &\text{ if }f(x)=1
        \end{cases}
    \end{align*}
    Given a family of partial functions $\mathcal{F}$, $\textup{Ada}_\mathcal{F} = \{\adaf{f}{d}: f\in\mathcal{F}, d\in\mathbb{N}\}$.
\end{definition}

The authors use $\adaf{f}{d}$ to show the following result.

\begin{theorem}[\cite{adaptivityvspostselection}]\label{theorem:pp_high_query_complexity}
    The query complexity of a \pp machine to compute $\adaf{\textup{AND}_n}{\log n}$ is $\Omega(\sqrt{n})$.
\end{theorem}

\section{Computational Models}\label{section:computational_models}

The models introduced in this paper have two motivations. Firstly, in~\cite{space_above_bqp}, the authors study the power of non-collapsing measurements (measurements which do not collapse the wavefunction) and its associated complexity class \pdqp. They show that $\pdqp \subseteq \compclass[\pp]{BPP}$, but it doesn't appear that $\pdqp \subseteq \pp$. \corrbqp was an attempt to bridge the gap between \pdqp and \pp by defining a metaphysical class which contains both. Secondly, \majbqp is motivated by addressing the computational power of always choosing the most-likely measurement outcome.

\subsection{Correlated Measurements}\label{subsection:correlated_measurements}

Before formally defining \corrbqp, let us define the correlated measurement gate.

\begin{definition}\label{definition_valid_correlated_outputs}
    Let $\rho$ be a state over the registers $R$. Suppose we partition the registers $R$ into $k$ equisized disjoint sets $R_1, ..., R_k$ where $\forall i, \abs{R_i}=j$. Additionally, for a constant $k$, let $M_{a,k} = (\ketbra{a}{a})^{\otimes k}$.
    
    The \emph{set of valid correlated outputs} $V$ is the set of strings $a\in\{0,1\}^j$ such that,
    \begin{align}
        \Tr (M_{a,k} \rho) > 0 \label{equation_valid_output}
    \end{align}
    Let $R_1$ be denoted as a \emph{leader} and let the state over registers $R_1$ be denoted by $\rho_1$. We call $\rho_1$ a \emph{leader state}. A correlated measurement $C_{\{R_i\}_{i\in[k]}, R_1}$ on $\rho$ is a measurement described by the collection of measurement operators $\{M_a^*\}_{a\in V}$ where,
    \begin{align}
        M_a^* = \frac{\Tr ( M_{a,1} \rho_1)}{\Tr (M_{a,k} \rho)} M_{a,k}\label{eq:correlated_measurement_operator_definition}
    \end{align}
    We call such operator a valid operator if $V\neq \emptyset$.
\end{definition}

Let us explain the definition of the correlated measurement operators $M_a^*$. The operator $M_{a,k}$ represents the condition that the values across all of ${R_i}_{i\in[k]}$ are equivalent. Therefore, presence of $M_{a,k}$ in the definition of $M_a^*$ in~\cref{eq:correlated_measurement_operator_definition} is represented by Step (1) of~\cref{fig:model_explanation}. Furthermore, the constant $\frac{\Tr ( M_{a,1} \rho_1)}{\Tr (M_{a,k} \rho)}$ in~\cref{eq:correlated_measurement_operator_definition} adjusts the weights in order to mimic the probability of measuring $a$ in the leader register, represented by Step (2) of~\cref{fig:model_explanation}. Lastly, we condition on the valid set of correlated outputs $V$ described by~\cref{equation_valid_output} as $\Tr (M_{a,k} \rho)=0$ would make $M_a^*$ undefined. Intuitively, it means that we may only collapse to an output which was already possible beforehand.

This means that applying the measurement described by $\{M_a^*\}_{a\in V}$ collapses all registers in $R$ to the same value $a\in V$ with probability,
\begin{align*}
    \frac{\bra{a}\rho_1\ket{a}}{\sum_{a^\prime \in V} \bra{a^\prime}\rho_1\ket{a^\prime}}
\end{align*}

The normalization term is necessary in case $V$ is a strict subset of $\{0,1\}^j$. Furthermore, note the assumption at the end of~\cref{definition_valid_correlated_outputs} that $V$ is not empty. If that was the case, we assume that the circuit is invalid.\footnote{The same condition holds for \postbqp. If for any input the postselected measurement outcome has zero probability, then the entire circuit is invalid.}
We emphasize that such operator does not satisfy the postulates of quantum mechanics. Specifically, the measurement operators do not satisfy the completeness equation,
\begin{align*}
    \sum_{a}  {M_a^*}^\dagger M_a^* \neq I
\end{align*}
However, as we are attempting to amplify the power of \bqp, such operation is necessary. Let us formally define the class \corrbqp.

\begin{definition}[\corrbqp]\label{definition_corrbqp}
    A language $L$ is said to be in \corrbqp if there exists a uniform family of quantum circuits $\{Q_n\}_{n\in \mathbb{Z}}$ composed of Hadamard, Toffoli and valid correlated measurement operators such that:
    \begin{itemize}
        \item If $x\in L$, the probability that $Q_n$ outputs $1$ to input $\ket{x}$ is at least $2/3$.
        \item If $x\not\in L$, the probability that $Q_n$ outputs $1$ to input $\ket{x}$ is at most $1/3$.
    \end{itemize}
\end{definition}

As we are using Hadamard and Toffoli gates, the condition for a valid correlated measurements in~\cref{equation_valid_output} becomes,
\begin{align*}
    \Tr ({M_{a,k}}^\dagger M_{a,k} \rho) \geq \frac{1}{2^h}
\end{align*}
where $h$ is the number of Hadamard gates used in $Q_n$.

\subsection{Quantum Majority}

Let us formally define the quantum majority gate, a single-qubit gate which collapses a state in the $\ket{0}, \ket{1}$ basis to the most-likely output.

\begin{definition}[Majority Gate]\label{definition:majority_gate}
    Let $\rho$ be a single-qubit state. Applying the \say{quantum majority} gate $M$ on $\rho$ results in the following:
    \begin{align*}
        M\rho = \begin{cases}
            \ket{0} &\text{if $\bra{0}\rho\ket{0} \geq 1/2$}\\
            \ket{1} &\text{otherwise}
        \end{cases}
    \end{align*} 
\end{definition}

Alternatively, we may view the gate above as postselecting on the more likely outcome between $\ket{0}$ and $\ket{1}$. Let us define \bqp enhanced with this gate.

\begin{definition}[\majbqp and \admajbqp]\label{definition:majbqp}
    A language $L$ is said to be in \majbqp if there exists a uniform family of quantum circuits $\{Q_n\}_{n\in \mathbb{Z}}$ composed of Hadamard, Toffoli and majority gates such that:
    \begin{itemize}
        \item If $x\in L$, the probability that $Q_n$ outputs $1$ to input $\ket{x}$ is at least $2/3$.
        \item If $x\not\in L$, the probability that $Q_n$ outputs $1$ to input $\ket{x}$ is at most $1/3$.
    \end{itemize}
    The class \admajbqp is defined similarly, except $Q_n$ also may perform \emph{intermediate measurements}.
\end{definition}

\section{Complexity Landscape}

\subsection{Basic Properties}

We establish some properties of \corrbqp and \majbqp which will be useful later.

\begin{lemma}\label{lemma_corrbqp_co}
    \corrbqp, \majbqp and \admajbqp are closed under complement.
\end{lemma}
\begin{proof}
    Let us describe the argument for \corrbqp; \majbqp and \admajbqp uses the same exact argument. For any language $L\in\corrbqp$, let $Q_n$ be the circuit per~\cref{definition_corrbqp} that recognizes $L$. Then for $\overline{L}$, the complement of $L$, we may decide whether $x\in \overline{L}$ by using the same circuit $Q_n$ and flipping the output register at the end. Therefore $\cocorrbqp\subseteq \corrbqp$. Showing the opposite direction uses the same argument.
\end{proof}

\begin{lemma}\label{lemma_corrbqp_error_amplification}
    $\corrbqpsc{2/3}{1/3} = \corrbqpsc{1-2^{-p(n)}}{2^{-p(n)}}$, $\majbqp (2/3, 1/3)\\ = \majbqp(1,0)$ and $\admajbqp(2/3,1/3) = \admajbqp(1-2^{-p(n)},2^{-p(n)})$.
\end{lemma}
\begin{proof}
    Let us begin with \corrbqp. In order to amplify the completeness and soundness gap, we perform $m$ computations in parallel and take a majority vote. By the Chernoff Bound~\cite{arora_barak}, this will decrease the error probability from $1/3$ to $2^{-e(m)}$ for some increasing function $e(m)$. What remains is setting $m$ such that $e(m)\geq p(n)$. A similar argument works for \admajbqp.

    For \majbqp, consider an arbitrary algorithm which outputs the correct answer with probability $2/3$. By using a majority gate at the end of the computation, we always output the correct answer.
\end{proof}

\begin{lemma}\label{lemma_corrbqp_selfcollapse}
    $\corrbqporacle{\corrbqp} = \corrbqp$, $\compclass[\majbqp]{MajBQP}_{\textup{classical}} = \majbqp$ and $\compclass[\admajbqp]{AdMajBQP}_{\textup{classical}} = \admajbqp$.
\end{lemma}
\begin{proof}
    Let us describe the argument for \corrbqp; \majbqp and \admajbqp uses the same exact argument. $\corrbqp\subseteq \corrbqporacle{\corrbqp}$ is straight-forward: The base machine does not make any queries to the oracle.
    Let us show the reverse. Let $Q_n$ be the \corrbqp circuit which includes polynomially many classical oracle calls to \corrbqp. Any classical call denotes that $Q_n$ has some classical query registers containing $x$ and a classical output register $o$. In the simulation using \corrbqp, we copy $x$ to a separate register, run a separate \corrbqp subroutine, measure the output at the end and copy the output to $o$. Using~\cref{lemma_corrbqp_error_amplification}, we can amplify the computation to obtain the wrong answer with exponentially low probability of error. Since the number of calls to the oracle is polynomial, the total accumulated error can still be bounded above any desired constant. Furthermore,~\cref{lemma_corrbqp_co} ensures that running the subroutine is sufficient for computing \emph{any} input. Finally, the resulting circuit $Q_n$ can also be amplified in order to ensure the error is at most $1/3$.
\end{proof}

Lastly, we show that correlated measurements may simulate a variety of other metaphysical operations.

\begin{lemma}\label{lemma:corr_simulate_postselection_cloning}
    Any quantum algorithm $A$ with correlated measurements may exactly simulate postselection and cloning states.
\end{lemma}
\begin{proof}
    \textbf{Postselection.} Without loss of generality, assume we are postselecting on register $r$ measuring $\ket{1}$. A correlated measurement may simulate this by initializing an ancilla to $\ket{1}$ and applying a correlated measurement on the ancilla and $r$.

    \textbf{Cloning.} Suppose we want to clone some state $\ket{\psi}$ which was created by applying a sequence of unitaries and partial measurements. In order to create $\ket{\psi}\otimes \ket{\psi}$, perform the same process on two registers simultaneously, replacing the partial measurements with correlated measurements. Due to the fact that the partial measurement probability is dictated by the leader register, we obtain exactly the same state as we would without the simulation.
\end{proof}

\subsection{Computing Parity}

It has been shown that when \bqp has the ability to clone states, it may solve search using a single query to the oracle and $O(\log n )$ cloning operations~\cite{grover_search_no_signaling_principle}.  We show that using correlated measurements, we may efficiently compute parity as well.

\begin{lemma}\label{lem:corrbqp_parity}
    The parity problem, defined as $x_1\oplus \dots \oplus x_n$ given some input $x\in\{0,1\}^n$, may be solved using $O(\log n)$ queries and $O(\log n)$ correlated measurements.
\end{lemma}
\begin{proof}
    The algorithm is as follows.
    \begin{enumerate}
        \item Prepare the state below
        \begin{align*}
            \ket{\psi_1} = \frac{1}{\sqrt{2n}} \sum_{i\in \{0,1\}^n} \ket{i}_O (\ket{0}_C + \ket{1}_C)
        \end{align*}
        \item Apply the oracle on register $O$, controlled on register $C=1$, obtaining 
        \begin{align*}
            \ket{\psi_2} = \frac{1}{\sqrt{2n}} \sum_{i\in \{0,1\}^n} \ket{i}_O (\ket{0}_C + (-1)^{x_i} \ket{1}_C)
        \end{align*}
        \item Repeat the following for $j$ from $0$ to $\log n$ (without loss of generality we assume $n$ is a power of $2$).
        \begin{enumerate}
            \item Let $\pi_j$ be a permutation on $[N]$ which maps $\pi_j(i) = \lfloor i/2^{j+1}\rfloor + (i + 2^{j}) \mod{2^{j+1}}$.
            \item Applying~\cref{lemma:corr_simulate_postselection_cloning}, simulate a copy operation on the current state of the algorithm $\ket{\psi^\prime}$, obtaining $\ket{\psi^\prime}\otimes \ket{\psi^\prime}$.
            \item Apply $\pi_j$ on the register $O$ of the first copy.
            \item Postselect on the instances where the registers $O$ are equal. Let $(-1)^{p_i}$ be the phase of $\ket{i}$ when $C=1$. After the postselection, this phase is $(-1)^{p_i \oplus \pi_j (p_i)}$.
            \item Discard one of the copies.
        \end{enumerate}
        \item After $\log n$ iterations, the register $C$ is in the state $\ket{0} + (-1)^{x_1 \oplus \dots \oplus x_n}\ket{1}$. By measuring in the Hadamard basis, we obtain the parity of the input.
    \end{enumerate}

    At each iteration over $j$, we perform one postselection and copy. The postselection step requires a single correlated measurement, while the copies only require simulating the previous correlated measurements. Therefore they may be incorporated in the correlated measurements simulating postselection, meaning we perform $O(\log n)$ correlated measurements. Additionally, in order to simulate copies, we need to make a query for each copy, costing us $O(\log n)$ oracle calls.
\end{proof}

\subsection{Characterizing metaphysical complexity classes}

We position the computational power of correlated measurements and quantum majority in the Complexity Zoo. Specifically, we show that $\corrbqp = \compclass[\pp]{BPP} = \admajbqp$, $\majbqp = \compclass[\pp]{P}$ and the Merlinizations of \corrbqp and \majbqp are equal to \compclass{NEXP}.

Let us begin with showing $\pdqp \subseteq \corrbqp$. Let us first describe the class \pdqp, a formal definition may be found in~\cite{space_above_bqp, revisiting_bqp_with_noncollapsing_measurements}. In summary, \pdqp is \bqp enhanced with the ability to obtain multiple samples of the computation states (so-called non-collapsing measurements) which are then processed using a deterministic Turing Machine.

\begin{lemma}\label{lemma_pdqp_in_corrbqp}
    $\pdqp\subseteq\corrbqp$
\end{lemma}
\begin{proof}
Any \pdqp algorithm may be simulated by an algorithm which clones its state prior to each non-collapsing measurement. Therefore, by~\cref{lemma:corr_simulate_postselection_cloning}, as \corrbqp may simulate cloning, we are done.
\end{proof}

We may similarly show that $\compclass{AdPDQP}\subseteq \corrbqp$.\footnote{\compclass{AdPDQP} is adaptive version of \pdqp as it may adapt computation based on previous non-collapsing measurements.} Additionally, due to the fact that correlated measurements may simulate non-collapsing measurements, the proofs which show $\compclass{PDQMA} = \compclass{NEXP}$~\cite{pdqma_nexp_1, pdqma_nexp_a}\footnote{\compclass{PDQMA} is \compclass{QMA} with non-collapsing measurements.} and $\compclass{PDQP/qpoly} = \compclass{ALL}$~\cite{pdqpqpoly} may be used to show that $\compclass{CorrQMA}=\compclass{NEXP}$ and $\compclass{CorrBQP/qpoly} = \compclass{ALL}$. Next, let us characterize \corrbqp.

\begin{lemma}\label{lemma_bpp_pp_in_corrbqp}
    $\compclass[\pp]{BPP}\subseteq \corrbqp$
\end{lemma}
\begin{proof}
    By definition, we have that $\compclass[\pp]{BPP}\subseteq \corrbqporacle{\pp}$ as we are only adding more power to the machine being run and \compclass{BPP} may only perform classical queries to the \pp oracle.

    By combining~\cref{lemma:corr_simulate_postselection_cloning} with the fact that $\postbqp = \pp$~\cite{pp_postbqp},  we find that $\pp \subseteq \corrbqp$ and thus $\corrbqporacle{\pp}\subseteq\corrbqporacle{\corrbqp}$. Finally, as $\corrbqporacle{\corrbqp} \subseteq \corrbqp$ by~\cref{lemma_corrbqp_selfcollapse}, we obtain the desired result.
\end{proof}

\begin{lemma}\label{lemma_corrbqp_in_bpp_pp}
    $\corrbqp\subseteq \compclass[\pp]{BPP}$
\end{lemma}
\begin{proof}
    We extend the argument from~\cite{space_above_bqp} which showed $\pdqp\subseteq \compclass[\pp]{BPP}$ for \corrbqp. Let us note that $\compclass[\pp]{BPP} = \compclass[\sharpp]{BPP}$ as repeated use of the \pp oracle may calculate the result of a problem in \sharpp.

    Let $Q_n$ be the quantum circuit run. For simplicity, suppose we only perform a single correlated measurement $C_{\{R_i\}_{i\in [k]}, R_1}$ where $\forall i, \abs{R_i} = l$. As all other operations are Hadamard and Toffoli gates, the probability of each measurement output is $\frac{2^c}{2^h}$ where we used $h$ Hadamard gates in total and $c\geq 0$. This holds even for probabilities which are conditioned on another measurement output. By the techniques used to show $\bqp\subseteq \pp$~\cite{adleman_demarrais_huang_quantum_computability}, we may ask the \sharpp oracle to compute these probabilities.

    Let $m_i$ denote the correlated output of the registers in $R_i$, $m=m_1 m_2\dots m_k$ and let $R_1$ be the registers over the leader state. We shall iteratively obtain $m_i(j)$ for $j\in[l]$, starting with $j=1$. First, we ask the \sharpp oracle for the probabilities $(P_{1,1}, ..., P_{1,k})$ where
    \begin{align*}
        P_{1,i} = \Pr[R_i(1) = 1]
    \end{align*}
    Suppose that there exists some $j\in[k]$ such that $P_{1,j}\in \{0,1\}$. If that is the case, we know that all valid correlated measurements must begin with the string $P_{1,j}$, meaning the \bpp machine may set $m_i(1) = P_{1,j}$ for all $i$. Note that it is not possible to have $j,n\in [k]$ such that $P_{1,j} = 1 - P_{1,n}$ as that would imply that there are no valid correlated measurement outputs, contradicting our assumption about the algorithm. On the other hand, suppose that $\forall i\in[k]$, $P_{1,i}\in (0,1)$. Then the \bpp machine flips a coin with bias $P_{1,1}$. Assuming it gets $a\in \{0,1\}$, $\forall i\in[k]$ it sets $m_i(1) = a$.

    Obtaining the remainder of $m$ continues analogously, except the \sharpp queries at step $j\in[l]$ for $j>1$ ask for the conditional probability of measuring $1$ conditioned on measuring $m_i(n)$ for $n\in[j-1]$. Formally, we ask for $(P_{j,1}, ..., P_(j,k))$ where
    \begin{align*}
        P_{j,i} = \Pr[R_i(j) = 1\;|\;R_1(1)R_1(2)\dots R_k(j-1) = m_1(1)m_1(2)\dots m_k(j-1)]
    \end{align*}
    
    The rest of processing $m_i(j)$ proceeds as described above. This algorithm may be extended to multiple correlated measurements by asking the \sharpp oracle for the probabilities conditioned on prior correlated measurement outputs.
\end{proof}

We proceed by characterizing \majbqp.

\begin{lemma}
    $\majbqp = \compclass[\pp]{P}$ and $\admajbqp = \compclass[\pp]{BPP} = \compclass[\pp]{BQP}_{\textup{classical}}$
\end{lemma}
\begin{proof}
    It has been shown that $\pp=\compclass{PQP}$~\cite{watrous_quantum_computational_complexity}, where \compclass{PQP} is the \bqp with unbounded error (as opposed to the standard polynomial-gap between acceptance and rejection probabilities in \bqp, there is no gap in \compclass{PQP}). Therefore, $\pp \subseteq \majbqp \subseteq \admajbqp$ as \majbqp may run the same circuit and apply a quantum majority gate to the output register at the end of the computation. 
    
    Let us show that $\compclass[\pp]{P} \subseteq \majbqp$ and $\compclass[\pp]{BQP}_{\textup{classical}} \subseteq \admajbqp$, both of which hold due to~\cref{lemma_corrbqp_selfcollapse}. In the first case, we have that $\compclass{P}\subseteq \majbqp$ relativizes, meaning that $\compclass[\pp]{P} \subseteq \compclass[\pp]{MajBQP} \subseteq \compclass[\majbqp]{MajBQP} = \majbqp$. In the second case, $\compclass{BQP}_{\textup{classical}}\subseteq \admajbqp_{\textup{classical}}$ relativizes as both may perform intermediate measurements. Therefore, $\compclass[\pp]{BQP}_{\textup{classical}} \subseteq\\ \compclass[\admajbqp]{AdMajBQP}_{\textup{classical}} = \admajbqp$.\footnote{The same argument does not show that $\compclass{BQP}_{textup{classical}}\subseteq \majbqp_{\textup{classical}}$ as \majbqp cannot perform intermediate measurements.}
    
    Next, we describe a simulation of \majbqp computation using $\compclass[\sharpp]{P} = \compclass[\pp]{P}$. Let $C$ be an arbitrary \majbqp circuit. Recall, as explained in previous sections, that given a circuit using a universal gate set (excluding the majority gate), we may ask \sharpp to describe the probability of outputting $\ket{1}$. \compclass[\pp]{P} may simulate \majbqp as follows.
    \begin{enumerate}
        \item Let $k$ be the number of majority gates in $C$. For $i\in [k]$, repeat the following:
        \begin{enumerate}
            \item Send the circuit $C$ until the application of the $i$-th majority gate, the previous majority gate outputs and the register we are applying the majority gate to the \pp oracle, asking for the output of the $i$-th gate conditioned on the previous outputs.
        \end{enumerate}
        \item Ask \pp for the output of the circuit conditioned on the majority gate outputs and return the result.
    \end{enumerate}

    Lastly, $\admajbqp \subseteq \compclass[\pp]{BPP}$ holds due to the same algorithm as above, except that every intermediate measurement may be replaced by a \sharpp query asking for the probability of outputting $1$, followed by a biased coin-flip.
\end{proof}

Lastly, let us show that correlated measurements possess the same computational power as other previously-studied metaphysical classes such as cloning, rewindable computation and adaptive postselection.

\begin{lemma}\label{theorem:corrbqp_equals_adpostbqp}
    \compclass{CorrBQP} = \compclass{AdPostBQP} = \compclass{RwBQP} = \compclass{CBQP}
\end{lemma}
\begin{proof}
    By~\cref{lemma_corrbqp_in_bpp_pp} and~\cref{theorem_adpostbqp_equals_rwbqp_equal_cqbp}, $\corrbqp \subseteq \compclass[\pp]{BPP} \subseteq \compclass{CBQP}$. On the other hand, by~\cref{lemma:corr_simulate_postselection_cloning}, correlated measurements may simulate cloning states, meaning that $\compclass{CBQP}\subseteq \corrbqp$. Thus, $\corrbqp = \compclass{CBQP}$, meaning that~\cref{theorem_adpostbqp_equals_rwbqp_equal_cqbp} implies the result.
\end{proof}

By combining all of the above, we obtain the following result. Note that as all of these equalities relativize as each containment followed using simulation.

\begin{theorem}\label{thm:equality_theorem}
    Relative to any classical oracle,
    $\corrbqp = \compclass[\pp]{BPP} = \admajbqp = \adpostbqp = \compclass{RwBQP} = \compclass{CBQP}$ and $\majbqp = \compclass[\pp]{P}$
\end{theorem}

By combining prior knowledge about \compclass{PP} and~\cref{thm:equality_theorem}, we obtain some implications.

Let $\corrbqp_{0} = \compclass{P}, \corrbqp_{1} = \corrbqp$ and $\corrbqp_{i} = \compclass[\corrbqp_{i-1}]{CorrBQP}$. Furthermore, define $\compclass[i]{PostBQP}_{\textup{classical}}$ analogously. 
\begin{corollary}\label{lemma_corrbqp_implications} The following statements hold:
    \begin{enumerate}
        \item\label{equation_pp_pp_equal_pp_corrbqp} $\compclass[\pp]{PP} = \compclass[\corrbqp]{PP}$
        \item\label{equation_ch_level_containment} For any $i\in \mathbb{N}$, $\countinghierarchylevel{i}\subseteq \corrbqp_{i} \subseteq \countinghierarchylevel{i+1}$.
        \item\label{equation_ch_equality} $\countinghierarchy = \cup_{i\in\mathbb{N}} \corrbqp_i$.
        \item \label{equation_postbqp_classical_hierarchy} For $i\in\mathbb{N}$, $\compclass[i]{PostBQP}_{\textup{classical}} \subseteq \corrbqp$.
        \item\label{equation_qcph_in_corrbqp_corrbqp} $\compclass{QCPH}\subseteq \compclass[\corrbqp]{CorrBQP}$
    \end{enumerate}
\end{corollary}
\begin{proof}
    Let us formally prove these results.~\cref{equation_pp_pp_equal_pp_corrbqp} holds as $\compclass[\pp]{PP} \subseteq \compclass[\corrbqp]{PP} \subseteq \compclass[\mathsf{P}^{\pp}]{PP} \subseteq  \compclass[\pp]{PP}$.

    \cref{equation_ch_level_containment} can be shown via induction. When $i=1$, we have that $\pp \subseteq \corrbqp = \compclass[\pp]{BPP} \subseteq \compclass[\pp]{PP}$. For the inductive argument, assume the statement holds for all integers less than $i$. Then $\countinghierarchylevel{i} = \compclass[\countinghierarchylevel{i-1}]{PP} \subseteq \compclass[\corrbqp_{i-1}]{PP} = \compclass[\corrbqp_{i-1}]{PostBQP} \subseteq \corrbqp_{i} = \compclass[\mathsf{PP}^{\corrbqp_{i-1}}]{P} = \compclass[\countinghierarchylevel{i}]{P} \subseteq \countinghierarchylevel{i+1}$.

    \cref{equation_ch_equality} comes directly from~\cref{equation_ch_level_containment}.

    \cref{equation_postbqp_classical_hierarchy} holds due to~\cref{lemma_corrbqp_selfcollapse} combined with~\cref{lemma:corr_simulate_postselection_cloning}.

    Lastly, it has been shown in~\cite{qph_original} that $\compclass{QCPH}\subseteq \compclass[\mathsf{PP}^{\compclass{PP}}]{P}$. This implies~\cref{equation_qcph_in_corrbqp_corrbqp} as $\compclass{QCPH}\subseteq \compclass[\mathsf{PP}^{\compclass{PP}}]{P} \subseteq \compclass[\mathsf{PP}^{\corrbqp}]{P} = \compclass[\corrbqp]{CorrBQP}$.
\end{proof}

Additionally, we obtain the following result about \countinghierarchy.

\begin{corollary}\label{lemma_ch_collapse_implications}
    If any of the following holds, then \countinghierarchy collapses to $\compclass[\pp]{BPP}$.
    \begin{enumerate}
        \item \label{case:ch_collapse_corrbqp_quantum}\corrbqp is self-low with respect to quantum queries (i.e. $\compclass[\corrbqp]{CorrBQP} = \corrbqp$).
        \item \label{case:ch_collapse_postbqp_classical}For any oracle $O$, $\compclass[O]{\pp}\subseteq\compclass[O]{CorrBQP}_{\textup{classical}}$.
    \end{enumerate}
    Furthermore, if \majbqp is self-low with respect to quantum queries (i.e. $\compclass[\majbqp]{\majbqp} = \majbqp$), $\countinghierarchy$ collapses to $\compclass[\pp]{P}$.
\end{corollary}
\begin{proof} 
    Consider~\cref{case:ch_collapse_corrbqp_quantum}, that $\compclass[\corrbqp]{CorrBQP} = \corrbqp$. This means that for all $i\geq 1$, $\corrbqp_{i} = \corrbqp$. Therefore, by~\cref{equation_ch_equality} from~\cref{lemma_corrbqp_implications}, $\countinghierarchy \subseteq \cup_{i\in\mathbb{N}} \corrbqp_i = \corrbqp = \compclass[\pp]{BPP}$. The \majbqp statement follows analogously.

    Assume~\cref{case:ch_collapse_postbqp_classical} holds. Therefore, for all $i\in\mathbb{N}$, $\countinghierarchylevel{i} \subseteq \compclass{CorrBQP}_{i, \textup{classical}} \subseteq \corrbqp$, so $\countinghierarchy \subseteq \compclass[\pp]{P}$.
\end{proof}

\subsection{Comments on the Self-low Property}\label{subsection_self_low}

Finally, one may ask themselves why \postbqp is not self-low with respect to classical queries. The main reasoning is due to the fact that \postbqp does not allow for intermediate measurements during its execution. As has been mentioned in prior work~\cite{pp_postbqp, rewindable_bqp}, proving the statement via simulation has the issue with correctly processing the error of the output. While amplification may bring the error of a subroutine simulating a query exponentially close to 0, we cannot guarantee it is 0. Therefore performing a query whose input depends on the result of a previous query must run the query in a superposition on both. However, the amplitude of the postselected output may be exponentially larger in the incorrect query than the correct query. This would mean that the incorrect computational path has a much larger amplitude than the correct one, making the simulation fail.

The next question one may ask is whether any of the classes such as \corrbqp or \adpostbqp are self-low with respect to queries in superposition. The standard attempt to perform simulations fails due to the inability of \adpostbqp to defer measurements until the end (if it could, $\adpostbqp=\postbqp$). However, if we are simulating a quantum query using a subroutine in superposition, the measurement output probabilities may significantly differ when conditioned on an input, but they must be the same across all of them. Secondly, the same issue which was described for \postbqp above remains for \adpostbqp, meaning that hoping to show $\compclass[\postbqp]{AdPostBQP} = \adpostbqp$ is tricky.

\section{Rational Trees and Postselection}

In this section we define the \emph{rational tree degree} measure and establish its relationship with the query complexity of \adpostbqp.

\subsection{Definitions}

\subsubsection{Rational Trees}

A rational tree $T_R$ is a binary tree where each node contains a rational function $R_i$. When evaluating such tree on an input $x\in \{0,1\}^n$, we flip a coin with bias $R_i(x)$. Therefore we assume that $R_i(x)\in [0,1]$ for all inputs. If we are at a leaf node, we output the output value of the coin, otherwise we go down the tree. The degree of the tree is the maximum degree of the product of nodes over one path in $T_R$.

We note that any such tree may be rewritten as a sum and product of rational functions. We may create such formula by processing each path separately by either adding a $R_i$ or $1-R_i$ term to its product at each node and taking a sum over all paths. We say that $T_R$ $\epsilon$-approximates a boolean function if for all $x$, $\abs{T_R(x) - f(x)}\leq \epsilon$. Let us formally define $\ratsumdeg$.

\begin{definition}\label{definition:rational_tree_degree}
    Consider a function $f:\{0,1\}^n \rightarrow \{0,1\}$. The rational tree degree up to $\epsilon$ error of $f$, denoted by $\ratsumdeg[\epsilon](f)$, is the minimum degree among all rational trees $T_R$ which $\epsilon$-approximate $f$.
\end{definition}

\subsubsection{Query Complexity of Adaptive Postselection}

An \adpostbqp algorithm is one whose postselection may depend on measurement results. Therefore, a quantum query algorithm with adaptive postselection is an algorithm which postselects on randomly-chosen outcome. Alternatively, one may view it as a tree where each node contains queries and postselection operators, and we split into two branches when a measurement is performed, each output weighted with their respective success probability. 
We denote the adaptive postselection query complexity of $f$ as $\adpostq{\epsilon}(f)$.

\subsection{Approximate Rational Tree}\label{section:ratsumdeg_adpostQ}

We show that the degree of a rational tree lower-bounds the query complexity of \adpostbqp.

\begin{theorem}\label{theorem_trdeg_lowerbounds_adpostq}
    For all $\epsilon\in [0,1/2)$ and $f:\{0,1\}^n\rightarrow \{0,1\}$,
    \begin{align*}
        \ratsumdeg[\epsilon](f) \leq 2\adpostq{\epsilon}(f)
    \end{align*}
\end{theorem}
\begin{proof}
    Let $T$ denote the number of queries made and $M$ the number of intermediate single-qubit measurements. We show that $\ratsumdeg[\epsilon](f) \leq 2\adpostq{\epsilon}(f)$ via induction on $M$. When $M=0$, the algorithm does not make any intermediate measurements, meaning it behaves as a \postbqp algorithm. Thus, per~\cite{MdW14}, the query complexity of the algorithm may be described using a rational function of degree at most $2T$. Specifically, one may define the polynomials $Q(x)=\Pr[p=1]$ and $P(x)=\Pr[p=o=1]$ which describe the query algorithm. Here $o$ is the output register and $p$ is the postselection register. Their ratio $P(x)/Q(x)$ equals $\Pr[o=1|p=1]$, the probability of outputting 1 when postselecting on 1. Per~\cite{quantum_lower_bounds_by_polynomials}, $P$ and $Q$ have degree at most $2T$, meaning that the degree of $P/Q$ is at most $2T$.
    
    Assume the statement holds for all integers less than $N>0$. Consider the first partial measurement is performed after $t<T$ queries on register $m$. Per the same argument as above, let $Q_1(x) = \Pr[p=1]$ and $P_{1,i} = \Pr[m=i \wedge p=1]$. Therefore $P_{1,i}/Q_1$ describes $\Pr[m=i| p=1]$ and has degree at most $2t$. At this point, the computation takes one of two paths based on $m$. Per the inductive argument, there exists a rational formula $R_i$ for each $i$ of degree at most $2(T-t)$. Thus the rational formula $R$ describing the query algorithm is $R = R_0(P_{1,0}/Q_1) + R_1(P_{1,1}/Q_1)$, may be described as a tree and has degree at most $2T$.    
\end{proof}

Unfortunately, we are not able to show $\adpostq{\epsilon}(f)\leq \ratsumdeg[\epsilon](f)$ for $\epsilon\neq 0$. The issue is that the algorithm in~\cite{MdW14} which was used to show $\postq{\epsilon}\leq \ratdeg[\epsilon]$ is biased towards outputting the correct answer. This causes an issue when evaluating multiple rational functions as it may add bias towards a path which outputs the incorrect output. We may extend prior results of~\cite{adaptivityvspostselection} to show a separation between $\ratdeg$ and $\ratsumdeg$.

\begin{lemma}
    Consider $f^\prime = \adaf{\textup{AND}_n}{\log n}$. Then $\ratdeg[1/3](f^\prime) = \Omega(\textup{poly}(n))$, while $\ratsumdeg[1/3](f^\prime)= O(\log ^3 n)$.
\end{lemma}
\begin{proof}
    By~\cref{theorem:pp_high_query_complexity} and~\cref{theorem:pp_equals_postbqp}, we know that $\postq{1/3}(f^\prime) = \textup{poly}(n)$. On the other hand, the algorithm in~\cite{pp_postbqp} which shows that \postbqp can compute majority may also be used to compute $\textup{AND}$ in the same time, namely $O(\log^2 n)$. Therefore \adpostbqp may run the \postbqp subroutine which computes each level of $\adaf{\textup{ADA}}{d}$. By measuring at the end of each subroutine, we ensure that the we obtain the correct output with high probability. As $\postq{1/3}(f^\prime)\leq \ratdeg[1/3](f^\prime)$ and $\ratsumdeg[1/3](f^\prime) \leq 2\adpostq{1/3}(f^\prime)$, we are done.
\end{proof}

We note that the original result of~\cite{adaptivityvspostselection} showed an oracle separation between \pp and \compclass[\compclass{NP}]{P}, implying an oracle separation between \postbqp and \adpostbqp.

\subsection{Exact Computation}

The tight relationship between $\ratdeg[\epsilon]$ and $\postq{\epsilon}$ established in~\cite{MdW14} suggests that the rational degree conjecture\footnote{The rational degree conjecture is asking whether in the zero-error setting, $\ratdeg(f)$ is polynomially related to $\widetilde{\text{deg}}(f)$ for all total boolean functions $f$~\cite{rational_degree_conjecture}. For the recent proof of the conjecture, see~\cite{KWZ26}.} may be approached by studying the complexity class \posteqp, \postbqp which outputs the correct answer with certainty. We show that adding adaptivity to \posteqp does not increase its power.

\begin{lemma}\label{lemma_posteqp_equals_adposteqp}
    For any classical oracle $O$,
    \begin{align*}
        \compclass[O]{PostEQP} =\compclass[O]{AdPostEQP}
    \end{align*}
\end{lemma}
\begin{proof}
    By definition, $\posteqp \subseteq \adposteqp$. Consider some language $L\in \adposteqp$, some input $x$ and an \adposteqp circuit $C$ as in~\cref{definition_circuit_postselect_measurements}. We know that if $x\in L$, $C$ always outputs 1 and always outputs 0 otherwise. Consider the first measurement $C$ performs. We know that regardless of the measurement output, we must always obtain the correct output. Therefore this measurement may be omitted as it does not alter the output probability. As the only difference between \posteqp and \adposteqp is the ability to perform partial measurements during computation, $\adposteqp \subseteq \posteqp$. Finally, this proof relativizes as any \posteqp circuit may simulate a \adposteqp circuit and vice-versa.
\end{proof}

We find this interesting as the same does not hold in the bounded-error case. Furthermore, when we consider the \emph{exact} version of other metaphysical complexity classes we study in this paper, we find that their relationships differ from their bounded-error counterparts.

\begin{observation}
    Given a bounded-error complexity class, let their exact version by denoted by replacing the symbol $\textup{B}$ which signifies bounded-error, with $\textup{E}$, signifying zero-error.\footnote{For complexity classes which do not contain the symbol $\textup{B}$, such as \pdqp, we add the $\text{E}$ in front of $\textup{Q}$.} We find that,
    \begin{align*}
        \eqp = \rweqp \subseteq \pdeqp \subseteq \ceqp \subseteq \adposteqp = \posteqp = \correqp \subseteq \majeqp
    \end{align*}
\end{observation}
\begin{proof}
    Firstly, \eqp is trivially contained by all classes above. Furthermore, $\rweqp \subseteq \eqp$. This is due to the fact that in exact quantum computation, regardless of a measurement outcome, we output the correct answer with certainty. Therefore, omitting the measurement also provides us with the correct output with certainty, making the rewinding operator superfluous.

    Secondly, per~\cref{lemma_posteqp_equals_adposteqp} and the fact that \correqp may simulate postselection, as shown in~\cref{lemma:corr_simulate_postselection_cloning}, $\posteqp=\adposteqp \subseteq \correqp$. On the other hand, $\correqp \subseteq \adposteqp$ due to the fact that regardless of the correlated measurement output, we must always output the correct one. Therefore, the amplitude adjustment based on the leader register does not affect the output, meaning we may instead simply postselect on the instances where all correlated registers are equal.

    Thirdly, $\eqp \subseteq \pdeqp \subseteq \ceqp \subseteq \correqp$ as copies may simulate non-collapsing measurements and \correqp may simulate copies. Lastly, $\majeqp = \majbqp$ by adding a quantum majority gate at the output register right before termination, we are done. 
\end{proof}

We find this interesting as a direct relationship between $\compclass{SZK}$ and \pp is not known,\footnote{There exist oracle separations in both directions, see~\cite{BCHTN16}.} the fact that $\compclass{SZK}\subseteq \pdqp$ and $\pdeqp \subseteq \pp$ shows another example of how allowing for bounded-error may amplify the power of a complexity class.
 
Lastly, when we consider the zero-error versions of \ratdeg and \ratsumdeg, we find that they are equivalent.

\begin{lemma} For any function $f:\{0,1\}^n\rightarrow \{0,1\}$,
    \begin{align*}
        \ratsumdeg[0](f) \leq \ratdeg[0](f) \leq 2\ratsumdeg[0](f)
    \end{align*}
\end{lemma}
\begin{proof}
    By definition, $\ratsumdeg[0](f) \leq \ratdeg[0](f)$. On the other hand, by~\cite{MdW14}, $\ratdeg[0](f) \leq 2\postq{0}(f)$. Let us show that $\postq{0}(f) \leq \ratsumdeg[0](f)$. Consider an arbitrary rational tree $T_R$ which exactly evaluates to $f$. Therefore, given an input $x$, any leaf which has non-zero probability of being evaluated must output $f(x)$. This means that without loss of generality, we may assume that any non-leaf node must also always output $1$ or $0$ for every $x$ as it suffices to describe one path to a leaf which outputs $f(x)$. This may be simulated in \postbqp by evaluating each node using the rational function evaluation algorithm, as described in~\cite{MdW14}, as it always outputs $f(x)$ when restricted to exact rational functions.
\end{proof}

Therefore, as $\ratdeg(f)$ is polynomially-related to deterministic query complexity $D(f)$ (and other complexity measures) for every total function~\cite{KWZ26}, so is $\ratsumdeg(f)$.

\section{Postselection with Classical Queries}\label{section:oracle_separations}

Finally, in order to address~\cref{lemma_ch_collapse_implications}, we study the class \postbqpclassical, \postbqp which may only make classical queries to the oracle. Of course, in the absence of an oracle, $\postbqpclassical = \pp$, but in the relativized world, we only know that for every classical oracle $O$, $\compclass[O]{PostBQP}_{\textup{classical}} \subseteq \compclass[O]{PP}$.

We extend the proof of~\cite{bqp_and_ph, bpppath_bqp_oracle_separation_fix} which separates \bqp from \postbpp to obtain a separation between \bqp (with \emph{quantum} query access) and \postbqpclassical. As $\bqp\subseteq \pp$ relativizes, there exists an oracle $O$ such that $\compclass[O]{PP}\not \subseteq \compclass[O]{PostBQP}_{\textup{classical}}$. Our result underscores the necessity to allow computational classes to posses full access to an oracle as even a slight restriction may thwart their computational power. This is interesting to see as \pp may only perform classical queries. We begin our proof with the proposition below.

\begin{lemma}\label{lemma:no_postbqpclassical_algorithm}
    Let $f: D\rightarrow \{0,1\}$ be a function where $D\subseteq \{0,1\}^n$. Suppose there are two distributions $\mathcal{D}_0, \mathcal{D}_1$ which are $\epsilon$-almost $k$-wise equivalent. If $\epsilon = o(1)$, then there does not exist a \postbqpclassical algorithm which may compute $f$ using $k$ queries.
\end{lemma}
\begin{proof}
    Let $M$ be some \postbqpclassical machine which computes $f$, meaning that it is a machine which makes classical queries to the input $x\in D$. Let $\Pi_p$ and $\Pi_o$ be the postselection and output projectors, respectively (we assume they are valid projectors for \postbqp, meaning that $\Pi_o \leq \Pi_p$). Note that these projectors are independent of the input $x$. Let $(i_1,\dots,i_k)$ be the queried indices of $x$ and $Z = (Z_1,\dots,Z_k) = (x_{i_1},\dots,x_{i_k})$ the outputs to the queries. After these queries, the circuits state $\rho_Z$ solely depends on $Z$. Given a distribution $D$ and $s\in \{p,o\}$, we define $p_{D,s}$ to be the probability that when $x$ is drawn from $D$, that $\Pi_s$ succeeds on $\rho_z$.
    \begin{align*}
         p_{D,s} \coloneqq \E_{x\sim D} \left[ \Tr(\Pi_s \rho_{Z(x)} ) \right] = \sum_{z\in \{0,1\}^k} \Pr_{x\sim D}[Z(x) = z] \Tr(\Pi_s \rho_z)
    \end{align*}
    Therefore the probability of accepting is $a_D = \frac{p_{D,o}}{p_{D,p}}$. This means that for $M$ to compute $f$, then $a_{D_0} \leq 1/3$, while $a_{D_1} \geq 2/3$.

    As $\mathcal{D}_0$ and $\mathcal{D}_1$ are $\epsilon$-almost $k$-wise equivalent, $(1-\epsilon)p_{D_0,s} \leq p_{D_1,s} \leq (1+\epsilon) p_{D_0,s}$, meaning that $1-\epsilon \leq p_{D_1,s}/p_{D_0,s} \leq 1+\epsilon$. Therefore,
    \begin{align*}
        1-\epsilon \leq \frac{p_{D_1,o}/p_{D_1,p}}{p_{D_0,o}/p_{D_0,p}} = \frac{a_{D_1}}{a_{D_0}} \leq 1+\epsilon
    \end{align*}
    If $M$ succeeds, we know that $\frac{a_{D_1}}{a_{D_0}} \geq 2$, but when $\epsilon = o(1)$, this is false.
\end{proof}

This suffices for us to prove our main result.

\begin{theorem}\label{theorem:separate_bqp_postbqp_classical}
    There exists an oracle $O$ such that $\compclass[O]{BQP}\not\subseteq \compclass[O]{PostBQP}_{\textup{Classical}}$.
\end{theorem}
\begin{proof}
    We may combine~\cref{theorem_forrelation_k_wise} with~\cref{lemma:no_postbqpclassical_algorithm} in the case that $k=\textup{polylog} (n) = n^{o(1)}$. Then, by standard diagonalization techniques, we obtain the oracle separation.
\end{proof}

Interestingly, it appears that even it might not hold that $\postbpp\subseteq \postbqpclassical$ with respect to all oracles, even though \postbpp may only make classical queries.

\begin{observation}\label{observation:implication_for_np_pp_in_bpp_pp}
    If $\compclass[O]{PostBPP}\subseteq \compclass[O]{PostBQP}_{\textup{classical}}$ with respect to any oracle $O$, then $\compclass[\pp]{NP} \subseteq \compclass[\pp]{BPP}$.
\end{observation}
\begin{proof}
    The proofs that $\compclass{NP}\subseteq\postbpp$ and $\postbqpclassical\subseteq \compclass{CorrBQP}_{\textup{classical}}$ relativize. Suppose that the proposition holds. Then we have that $\compclass[O]{NP}\subseteq\compclass[O]{CorrBQP}_{\textup{classical}}$ with respect to any oracle $O$, meaning that $\compclass[\pp]{NP}\subseteq\compclass[\pp]{CorrBQP}_{\textup{classical}}\subseteq \compclass[\compclass{CorrBQP}]{CorrBQP}_{\textup{classical}} = \corrbqp = \compclass[\pp]{BPP}$.
\end{proof}

Intuitively, the reason why $\compclass[O]{PostBPP}\subseteq \compclass[O]{PostBQP}_{\textup{classical}}$ is not straightforward is due to the fact that the postselection in \postbpp is based on a machine outputting $1$, given a random string as an input (i.e. we are postselecting on the strings for which the machine outputs 1). In this case, that machine may make classical queries to the oracle. If we try to simulate this postselection using \postbqp, we must run a computation with all the random strings in a superposition. However, that means that each query may be conditioned on the random string we are simulating, meaning that it would have to be in superposition. Finally, note that in~\cref{observation:implication_for_np_pp_in_bpp_pp}, the assumption could be weakened to the question of whether $\compclass{NP}\subseteq \compclass{CorrBQP}_{\textup{classical}}$ relativizes.

\section{Adaptive Non-Collapsing Measurements}

In~\cite{revisiting_bqp_with_noncollapsing_measurements}, the authors extended the Adversary lower-bounding technique~\cite{adversary_original} to \pdqp, quantum computation modified with the ability to perform non-collapsing measurements. As a common theme in our paper is the role of adaptivity, we study \adpdqp, a version of \pdqp which may adapt its computation based on the non-collapsing measurement outputs, by extending the Adversary lower-bounding technique to it. As we are introducing a version of the model which was not previously studied, let us formally define it.

\subsection{Definition}

\cite{space_above_bqp} defined the class \pdqp as follows. A polynomial-time Turing machine $T$ is given access to an oracle $O$ which takes as an input a list of unitaries $\{U_i\}_{i\in [k]}$ and measurement operators $\{M_i\}_{i\in [k]}$. $T$ makes a \emph{single} query the oracle $O$, which runs the sequence of unitaries and measurement operators. At each step $i$, the oracle obtains an independent sample $s_i$ distributed according to the measurement probabilities of the current state of the algorithm $\ket{\psi_i}$. At the end of the computation, the oracle returns the measurement results $\{s_i\}_{i\in [k]}$, which $T$ processes and outputs. Note that the power of the model lies in the ability to perform a (collapsing) partial measurement, followed by non-collapsing measurements (without the partial measurement operators $M_i$, this process can be simulated by \bqp).

\begin{figure}
    \centering
    \input{TikZ/adpdqp_original}
    \caption{Visualization of an \adpdqp algorithm. If one removed the \textcolor{orange}{orange} lines (i.e. the effect of $v_1$ and $v_2$), it would become a \pdqp algorithm.}
    \label{fig:adpdqp_original}
\end{figure}

Furthermore,~\cite{space_above_bqp} introduced an adaptive version of \pdqp which they call \compclass{CQP}. The difference is that $T$ machine may make \emph{multiple adaptive} queries the oracle $O$. Therefore, based on the previous non-collapsing measurements, they may adjust the unitaries sent to $O$. We define a stronger adaptive version of \pdqp, which we call \adpdqp. When the oracle returns the non-collapsing measurement samples, we allow it to retain the current state $\ket{\psi}$ for the next round of computation. From the oracle's point of view, it may ask $T$ which unitary it may apply based on the latest non-collapsing measurement result. A visualization of this process may be found in~\cref{fig:adpdqp_original}. Let us formally define the class.

\begin{definition}[\adpdqp]\label{def:adpdqp}
    A language $L$ is said to be in \adpdqp if there exists a uniform polynomial-time Turing machine $T$ with \emph{adaptive} access to the oracle $O$ which takes as an input a quantum state $\ket{\psi}$ and a description of a quantum algorithm $A = (U_1, M_1,\dots , U_{\textup{poly}(n)})$, runs the algorithm on $\ket{\psi}$ and returns the resulting state along with non-collapsing measurement samples obtained after each unitary, such that,
    \begin{itemize}
        \item If $x\in L$, then $T^O(x)$ accepts with probability at least $2/3$
        \item If $x\not\in L$, then $T^O(x)$ accepts with probability at most $1/3$
    \end{itemize}
\end{definition}

Note that in the definition of the oracle $O$ above, since $T$ is classical, it cannot interact with the output state - it may solely provide it to the oracle. When no state is given, $O$ assumes that the input is $\ket{0}$. Therefore, $\compclass{CQP} \subseteq \adpdqp$. Furthermore, without loss of generality, we may assume that each query to $O$ only contains a single unitary and and measurement operator, as is visualized in~\cref{fig:adpdqp_original}.

We provide an alternative formulation of the \adpdqp model. Before we do so, let us explain the alternative formulation proposed in~\cite{revisiting_bqp_with_noncollapsing_measurements}.
Specifically, one may view \pdqp through the lens of correlated measurements. Suppose some \pdqp algorithm $A$ performs $P$ non-collapsing measurements. We may simulate each non-collapsing measurement separately by applying the same set of unitaries in parallel. However, $A$ also contains partial measurements $M_i$. In order to simulate the model exactly, we must ensure that across each copy simulating a non-collapsing measurement $M_i$ collapses to the same value. This is exactly where correlated measurements may be used. By applying a correlated measurement on each copy where the main computation is the leader state, we may simulate the \pdqp computation exactly. This is shown in~\cref{fig:adpdqp_simulation} where the $C^*_i$ gates represent a correlated measurement with register $3$ being the leader.

With a slight tweak, we may do so with \adpdqp as well. Suppose that we performed a non-collapsing measurement, meaning that we measured the register $r$ simulating the non-collapsing measurement output to some value $a$. Therefore, the next unitary may depend on $a$. Let us denote it $U_a$. Therefore, when running the simulation above, each copy of the algorithm may apply a single unitary which, controlled on $r$ being $a$, applies $U_a$, over all possible outputs $a$. The orange lines in~\cref{fig:adpdqp_simulation} represent this modification. This formulation is crucial for the proof of our lower-bounding technique. As we are using fidelity as the progress measure, we may use the facts that fidelity is invariant under the application of unitaries and measurement cannot decrease fidelity. On the other hand, as $r$ is solely used as a control register, we never entangle other simulations of non-collapsing measurements, meaning that we may use most of the same tools used in~\cite{revisiting_bqp_with_noncollapsing_measurements} which relied on unentangled copies. This emphasizes the difference between \adpdqp and \compclass{CBQP} for which~\cite{revisiting_bqp_with_noncollapsing_measurements} obtained an exponentially-weaker bound exactly due to the fact that copies could become entangled.

\begin{figure}
    \centering
    \input{TikZ/adpdqp_simulation}
    \caption{Visualization of our formulation of a \adpdqp algorithm. The orange lines represent conditioning based on the non-collapsing measurements. Note that the $C^*_i$ measurement operators are the correlated measurements where the top register is the leader. Without the \textcolor{orange}{orange} lines, it would become a \pdqp algorithm.}
    \label{fig:adpdqp_simulation}
\end{figure}

\subsection{Adversary Lower-Bound}

Let us prove the lower-bounding technique for \adpdqp. Before we dive into the statement, let us emphasize that the query complexity of any version of \pdqp must measure both the number of queries $Q$ and non-collapsing measurements $P$. If this was not the case, then any query problem could be trivially solved by making a single query in superposition, followed by repeating non-collapsing measurements until we learned the entire function.

\begin{theorem}\label{thm:adversary_bound_adpdqp}
    Consider a function $f:\{0,1\}^n\rightarrow \{0,1\}$ and the sets $X\subseteq f^{-1}(0)$, $Y\subseteq f^{-1}(1)$. Given a fixed relation $R\subseteq X\times Y$, define the variables $m,l$ as,
    \begin{align*}
        m &= \min_{x\in X} \abs{\{y\in Y: (x,y)\in R \}} \\
        l &= \max_{\substack{x\in X\\ i\in [n]}} \abs{\{y\in Y: (x,y)\in R\text{ and } x(i)\neq y(i)\}}
    \end{align*}
    Define $m^\prime, l^\prime$ similarly, except $x\in X$ and $y\in Y$ are interchanged. For any \adpdqp algorithm which makes $Q$ queries and $P$ non-collapsing measurements,
    \begin{align*}
        QP = \Omega \left(\sqrt{\frac{mm^\prime}{ll^\prime}} \right)
    \end{align*}
\end{theorem}
Before we prove~\cref{thm:adversary_bound_adpdqp}, let us prepare several key statements. First, we will use the following statement which bounds the effect of a parallel query across copies of states on the fidelity between the states.

\begin{lemma}[Proof of Lemma 4.2 in~\cite{revisiting_bqp_with_noncollapsing_measurements}]\label{lem:bound_copies}
    Suppose that you are given oracle access to inputs $x,y \in \{0,1\}^n$. Describe the current states of the algorithm as $\ket{\psi_x} = \sum_j \alpha_{x,j} \ket{j}_q\otimes \ket{\phi_{x,j}}_w$. Let $\rho_x = \ket{\psi_x}^{\otimes c}$ and $\rho_x^\prime$ be the application $O_x$ to all $c$ query registers $q$ in parallel in $\rho_x$. Then,
    \begin{align*}
        F(\rho_x,\rho_y) - F(\rho_x^\prime, \rho_y^\prime) \leq 2c \sum_{j: x(j)\neq y(j)} \abs{\alpha_{x,j}}\abs{\alpha_{y,j}}
    \end{align*}
\end{lemma}

Furthermore, we will be using the variable $\Phi(t)$ to describe the progress made after $t$ queries, defined as,
\begin{align*}
    \Phi(t) = \sum_{(x,y)\in R} F(\rho_{x,t}, \rho_{y,t})
\end{align*}
where $\rho_{x,t}$ is the state of the algorithm after $t$ queries when given oracle access to the input $x\in \{0,1\}^n$ and $F$ denotes the fidelity measure, defined as $F(a,b) = \|\sqrt{a} \sqrt{b}\|_1$ where $a,b$ are density matrices. We will use the following statement which relates $\Phi(t)$ with the variables $m,m^\prime, l,l^\prime$.

\begin{claim}[Proof of Lemma 4 in~\cite{Amb04}]\label{clm:weight_on_query}
    Let $\beta_{x,t,i}$ be the amplitude weight on querying $j$ at some step $t$ of the algorithm. Then,
    \begin{align*}
        \sum_{\substack{(x,y)\in R\\ j: x(j)\neq y(j)}} 2 \beta_{x,t,j} \beta_{y,t,j} \leq \sqrt{\frac{l l^\prime}{m m^\prime}} \abs{R}
    \end{align*}
\end{claim}

Lastly, let us mention several results about the fidelity measure $F$.

\begin{claim}[Lemma 3.5 in~\cite{revisiting_bqp_with_noncollapsing_measurements}]\label{clm:fidelity_correlated_measurements}
    Let $\rho_x, \rho_y$ be to two states and $\rho_x^\prime, \rho_y^\prime$ be $\rho_x, \rho_y$ after applying a correlated measurement on $c$ copies of $\rho_x, \rho_y$. Then,
    \begin{align*}
        F((\rho_x)^{\otimes c}, (\rho_y)^{\otimes c}) \leq F((\rho_x^\prime)^{\otimes c}, (\rho_y^\prime)^{\otimes c})
    \end{align*}
\end{claim}

Note that~\cite{revisiting_bqp_with_noncollapsing_measurements} does not mention correlated measurements, but the action $M^*$ described in the original statement is a correlated measurement over exact copies of states.

\begin{claim}\label{clm:fidelity_on_cq}
    Suppose that we are given two mixed states $\rho_x, \rho_y$ which may be written as $\rho_x = \sum_{\tau} p_{x,\tau} \ketbra{\tau}{\tau} \otimes \sigma_{x,\tau}$ where $\tau$ is an orthonormal basis. Then,
    \begin{align*}
        F(\rho_x, \rho_y) = \sum_{\tau} \sqrt{p_{x,\tau} p_{y,\tau}} F(\sigma_{x,\tau}, \sigma_{y,\tau})
    \end{align*}
\end{claim}
\begin{proof}
    Recall that $F(a,b) = \|\sqrt{a} \sqrt{b}\|_1$. As $\sqrt{\rho_x} = \sum_\tau \sqrt{p_{x,\tau}} \ketbra{\tau}{\tau} \otimes \sqrt{\sigma_{x,\tau}}$, we have that,
    \begin{align*}
        F(\rho_x, \rho_y) &= \norm{\sqrt{\rho_x} \sqrt{\rho_y}}_1 = \sum_{\tau} \sqrt{p_{x,\tau} p_{y,\tau}} \norm{\sqrt{\sigma_{x,\tau}}\sqrt{\sigma_{y,\tau}} }_1\\
        &= \sum_{\tau} \sqrt{p_{x,\tau} p_{y,\tau}} F(\sigma_{x,\tau}, \sigma_{y,\tau})\qedhere
    \end{align*}
\end{proof}

We may proceed with the proof.

\begin{proof}[Proof of~\cref{thm:adversary_bound_adpdqp}]
    Consider an arbitrary \adpdqp algorithm $A$. After each non-collapsing measurement, we may update the computation based on the result. To track the changes, for $i\in [P]$, let $\tau_i$ denote the current transcript of non-collapsing measurements made. We are going to assume that right before making query $t$, we have made $i$ non-collapsing measurements. Therefore, we may describe the state of the algorithm before query $t$ as,
    \begin{align*}
        \rho_{x,t} = \sum_{\tau_i \in \{0,1\}^{ir}} p_{x,t}(\tau_i) \ketbra{\tau_i}{\tau_i} \otimes \sigma_{x,t,\tau_i}
    \end{align*}
    where $r$ is the number of registers and $p_{x,t}(\tau_i)$ denotes the probability of obtaining the transcript $\tau_i$ after $t$ queries to $O_x$. Let $\Phi(t)$ be the progress measure, defined as,
    \begin{align*}
        \Phi(t) &= \sum_{(x,y)\in R} F\left(\rho_{x,t}, \rho_{y,t}\right)\\
            &= \sum_{(x,y)\in R} \sum_{\tau_i} \sqrt{p_{x,t}(\tau_i) p_{y,t}(\tau_i)} F\left(\sigma_{x,t,\tau_i}, \sigma_{y,t,\tau_i}\right)\numberthis\label{eq:fidelity_cq_bound}
    \end{align*}
    where~\cref{eq:fidelity_cq_bound} uses~\cref{clm:fidelity_on_cq}. Furthermore, by~\cref{clm:fidelity_correlated_measurements}, if we condition on some transcript $\tau_i$, we may assume that the current state $\sigma_{x,t,\tau_i}$ is a pure state composed of $P-i$ copies of the same state. We may describe a single copy of such state as,
    \begin{align*}
        \ket{\psi_{x,t}^{\tau_i}} = \sum_j \alpha_{x,t}^{\tau_i}(j)\ket{j}_q \ket{\phi_{x,t}^{\tau_i}(j)}_w
    \end{align*}
    where $q$ denotes the query register and $w$ the workspace register. Let $\Tilde{\alpha}_{x,t} (\tau_i, j) = \sqrt{p_{x,t}(\tau_i)} \alpha_{x,t}^{\tau_i}(j)$. By adjusting the amplitudes of $\ket{\psi_{x,t}^{\tau_i}}$ with the weight of the transcript $\tau_i$, using~\cref{lem:bound_copies} and~\cref{eq:fidelity_cq_bound},
    \begin{align*}
        \Phi(t-1) -\Phi(t) &\leq 2(P-i) \sum_{\substack{(x,y)\in R\\j: x(j)\neq y(j)}} \sum_{\tau_i} \abs{\Tilde{\alpha}_{x,t} (\tau_i, j)} \abs{\Tilde{\alpha}_{y,t} (\tau_i, j)}\\
        &\leq 2P \sum_{\substack{(x,y)\in R\\j: x(j)\neq y(j)}} \sqrt{\sum_{\tau_i} \abs{\Tilde{\alpha}_{x,t} (\tau_i, j)}^2} \sqrt{\sum_{\tau_i} \abs{\Tilde{\alpha}_{y,t} (\tau_i, j)}^2}\numberthis\label{eq:cauchy_schwartz}\\
        &=  2P \sum_{\substack{(x,y)\in R\\j: x(j)\neq y(j)}} \beta_{x,t}(i) \beta_{y,t}(i)\numberthis\label{eq:label_cs}\\
        &\leq P \sqrt{\frac{l l^\prime}{m m^\prime}} \abs{R}\numberthis\label{eq:bound_on_query_weight}
    \end{align*}
    where~\cref{eq:cauchy_schwartz} is an application of the Cauchy-Schwartz inequality and~\cref{eq:label_cs} is a change of variables $\beta_{x,t}(i) = \sqrt{\sum_{\tau_i} \abs{\Tilde{\alpha}_{x,t} (\tau_i, j)}^2}$. Notice that $\beta_{x,t}(i)^2$ is the probability of measuring $i$ in the query register $q$. Therefore,~\cref{eq:bound_on_query_weight} follows from~\cref{clm:weight_on_query}. Thus, we have that,
    \begin{align*}
        \Phi(0)-\Phi(Q) \leq \sum_{t=1}^{Q}(\Phi(t-1) - \Phi(t)) \leq PQ \sqrt{\frac{l l^\prime}{m m^\prime}} \abs{R}\numberthis\label{eq:telescope}
    \end{align*}

    Assume that the success probability is $1-\epsilon$. In order to solve the problem successfully, we require $\Phi(Q) \leq 2\sqrt{\epsilon (1-\epsilon}\abs{R}$. With $\epsilon = 1/3$, as $\Phi(0) = \abs{R}$, this means that $\Phi(0) - \Phi(Q) \geq \Omega(\abs{R})$. Combined with`\cref{eq:telescope}, this implies $QP = \Omega (\sqrt{mm^\prime/ll^\prime})$.
\end{proof}

\cref{thm:adversary_bound_adpdqp} exactly extends the adversary method for \pdqp, meaning we obtain the same lower-bounds (see Section 7.1 in~\cite{revisiting_bqp_with_noncollapsing_measurements}). For example, we have a $\Omega(n^{1/4})$ lower-bound on unstructured search, $\Omega(n^{1/2})$ lower-bound on both parity and majority, and $\Omega(n^{1/4})$ lower-bound on the element distinctness problem. Therefore, we have oracle separations between \adpdqp and \compclass{NP}, \compclass{PP}, and \compclass{\oplus P}.
Furthermore, note that~\cref{thm:adversary_bound_adpdqp} may be lifted to the stronger positive-weighted adversary method by introducing a weight coefficient in the progress measure $\Phi(t)$.

\section{Discussion}

There is a variety of open questions in this direction, most of which stem from our lack of understanding of the space between \pp and \pspace. Let us list several of them.

\begin{itemize}
    \item \cref{thm:equality_theorem} exactly characterizes \compclass[\pp]{P} and \compclass[\pp]{BPP} using \majbqp and \admajbqp. Is there a way to \say{derandomize} \compclass[\pp]{BPP} to \compclass[\pp]{P}? Alternatively, is it possible to separate them with respect to an oracle?
    \item We find that both search and parity may be solved efficiently using correlated measurements. Which problems have a high query complexity in this model? Finding such problems could provide us the largest class with a classical oracle separation from \pspace (At the moment it is $\compclass[\compclass{PH}]{PP}$~\cite{parity_p_oracle_pp_ph}).
    \item Are there other metaphysical modifications of \bqp which are equal to different classes between \pp and \pspace, such as $\compclass[\pp]{BQP}$?
    \item Are any of such modifications mentioned above self-low with respect to quantum queries? As mentioned in~\cref{lemma_ch_collapse_implications}, this would imply the collapse of \countinghierarchy.
    \item Is there a statement which would address the power of \compclass{PP/poly} in a similar way to how the Karp-Lipton addresses \compclass{P/poly}~\cite{karp_lipton}?
    \item Can \corrbqp be used to strengthen the current quantum-classical version of Toda's theorem, $\compclass{QCPH}\subseteq \compclass[\mathsf{PP}^{\mathsf{PP}}]{P}$~\cite{qph_original}?
    \item Does $\adpdqp = \pdqp$? Alternatively, does there exist an oracle separating them? As was mentioned in~\cite{space_above_bqp}, \adpdqp has the ability to clone states up to polynomial-accuracy, whereas \pdqp does not. Does this aid the computation?
\end{itemize}

\printbibliography

\end{document}